\documentclass[10pt,a4paper]{article}
\usepackage[utf8]{inputenc}
\usepackage[english]{babel}
\usepackage{hyperref}
\usepackage{amsmath,amsthm}
\usepackage{amssymb}
\usepackage{graphicx}
\usepackage{tikz}
\usetikzlibrary{calc}
\usetikzlibrary{decorations.pathmorphing}
\usepackage{subfigure}
\usepackage{refcount}
\usepackage[hmargin=2.5cm,vmargin=3cm]{geometry}
\usepackage[color=green!30]{todonotes}
\usepackage{booktabs}

\newtheorem{theorem}{Theorem}

\newtheorem{proposition}[theorem]{Proposition}
\newtheorem{lemma}[theorem]{Lemma}
\newtheorem{observation}[theorem]{Observation}
\newtheorem{corollary}[theorem]{Corollary}

\newtheorem{definition}[theorem]{Definition}

\newcommand{\nto}{\ensuremath{\stackrel{n}{\longrightarrow}}}

\newcommand{\rgphom}{\textsc{RGPHom}}
\newcommand{\HOM}{\textsc{Hom}}

\newcommand{\decisionpb}[4]{
\begin{center}
        \noindent\framebox{\begin{minipage}{#4\textwidth}
                #1\\
                Instance: #2\\ 
                Question: #3
        \end{minipage}}
\end{center}
}

\title{Complexity of conjunctive regular path query homomorphisms\footnote{This research was financed by the French government IDEX-ISITE initiative 16-IDEX-0001 (CAP 20-25), the International Research Center “Innovation Transportation and Production Systems” of the I-SITE CAP 20-25, the CNRS IRL ReLaX, the IFCAM project ``Applications of graph homomorphisms'' (MA/IFCAM/18/39), and the ANR project GRALMECO (ANR-21-CE48-0004). An extended abstract (without the fifth author) of this paper was published in the proceedings of the CIE'19 conference~\cite{confversion}, however the proof of the main result was incorrect and is corrected here (with a weaker statement).}}
\author{Laurent Beaudou\footnote{Université Clermont Auvergne, CNRS, Clermont Auvergne INP, Mines Saint-\'Etienne, LIMOS, 63000 Clermont-Ferrand, France. E-mails: laurent.beaudou@uca.fr, florent.foucaud@uca.fr, lhouari.nourine@uca.fr}
\and Florent Foucaud\footnotemark[2]
\and Florent R. Madelaine\footnote{LACL, Université Paris-Est Créteil, Créteil, France. E-mail: florent.madelaine@u-pec.fr}
\and Lhouari Nourine\footnotemark[2]
\and Anantha Padmanabha\footnote{Indian Institute of Technology Madras, Chennai, India. E-mail: ananthap@cse.iitm.ac.in}
\and Ga\'etan Richard\footnote{Normandie Univ, UNICAEN, ENSICAEN, CNRS, GREYC, Caen (France). E-mail: gaetan.richard@unicaen.fr}}

\begin{document}

\maketitle

\begin{abstract}
A graph database is a digraph whose arcs are labeled with symbols from a fixed alphabet. A regular graph pattern (RGP) is a digraph whose arcs are labeled with regular expressions over the alphabet. RGPs model navigational queries for graph databases called \emph{conjunctive regular path queries} (CRPQs). A match of a CRPQ in the database is witnessed by a special navigational homomorphism of the corresponding RGP to the database, where vertices of the RGP are mapped by a function $f$ to the vertices of the database, and for each arc $uv$ of the RGP, we require the existence of a walk in the database from $f(u)$ to $f(v)$, such that the word corresponding to the sequence of labels of that walk belongs to the regular language labeling the arc $uv$. 

We extend this definition to one of navigational homomorphisms between two RGPs, modeling a homomorphism-based form of CRPQ containment. We study the complexity of deciding the existence of such a homomorphism between RGPs. We show that this problem, in the case where general walks are considered, can still be solved by an EXPSPACE algorithm (just like general CRPQ Containment, which is known to be EXPSPACE-complete). In the case of simple paths or trails, in fact the navigational homomorphism problem considered here becomes much more tractable, as it is PSPACE-complete. We also study the problem for restricted RGPs over a unary alphabet, that arise from some applications like XPath or SPARQL. For this case, homomorphism-based CRPQ containment is in NP. We prove that certain interesting cases are in fact polynomial-time solvable.
\end{abstract}

\section{Introduction}

Graphs are a fundamental way to store and organize data. Graph based data management systems have been an active area of research for more than three decades \cite{Y90} and has seen an increased interest both in academic research and in the industry \cite{B13}. A graph database can be seen as a directed graph with arc-labels (possibly also vertex-labels). Various query languages have been introduced to retrieve answers from graph databases. See for example the graph query languages G-CORE~\cite{GCORE} and CYPHER~\cite{CYPHER} for graph databases, with one of the earliest ones being G~\cite{CMW87}. In fact, an ISO standard has been published in April 2024, see \cite{GQL-ISO} for key highlights. 

Regular path queries (RPQs) and Conjunctive Regular Path queries (CRPQs) are generalizations of conjunctive queries (CQs) that form the basis of querying graph databases \cite{B13, AAB17, FLS98, MW95}. CRPQs can be seen as graphs where vertices are variables and the arcs are labeled by regular expressions.
Classically, evaluating CQs for graph databases can be modeled as graph homomorphisms~\cite{KV00,morvanPHD}. In this setting, a query is itself a graph, and a match is modeled by a homomorphism of the query graph to the database graph, that is, a vertex-mapping that preserves the graph adjacencies and labels. Thus, the CQ evaluation problem has close correspondence with the constraint satisfaction problem (CSP)~\cite{Vardi2000}, that is also expressed as a homomorphism problem. Classical homomorphisms are often not powerful enough to be used for evaluation of CRPQs. This is because homomorphisms are inherently local, whereas CRPQs allow non-local pattern matching by looking for arbitrary-length paths instead of arcs. With the increase in graph database applications in real-world tools, it is important to study static analysis problems for CRPQs which will lead to developing optimization techniques for CRPQ evaluation. Consequently, there have been many works that look at the expressive power~\cite{FS13}, containment problem~\cite{Figueira20,FGKMNT20}, boundedness problem~\cite{BFR19, FKM24}, minimization~\cite{FMR25} of CRPQs, among others.

The query containment problem is one of the fundamental problems in static analysis of queries.
Given two boolean CRPQs $P$ and $Q$, we say that $P \subseteq Q$ if for every graph database instance $D$, if $P$ evaluates to {\tt true} in $D$, then $Q$ also evaluates to {\tt true} in $D$. 
Note that containment is a purely semantic notion and does not give us any insight into the structural or syntactic relationship between $P$ and $Q$. In fact, for CQs, there is a classical result \cite{CH82} that states that if $P$ and $Q$ are CQs then $P \subseteq Q$ if and only if there is a homomorphism from $Q$ to $P$. Although~\cite{Figueira20} proves that the containment problem for CRPQs is EXPSPACE-complete, we do not have an analogous structural characterization result for CRPQ containment. There are two main technical difficulties while defining a structural equivalent for the containment in CRPQs: $(i)$ handling non-local pattern matching by looking for arbitrary length paths instead of arcs; and $(ii)$ handling regular expressions as arc labels.
A first step towards this is considered in~\cite{BarceloVardi2017}, where the authors define a notion of homomorphisms over \emph{regular graph patterns} (RGPs). 

In this paper, we take a further step in this direction and consider the notion of \emph{navigational homomorphism} which is a syntactic generalization of the classical homomorphism. Indeed, for CRPQs $P$ and $Q$, the existence of a navigational homomorphism from $Q$ to $P$ will imply that $P \subseteq Q$. We prove that checking for the existence of a navigational homomorphism is in EXPSPACE. Motivated by various applications, we further identify some restricted classes of queries for which we prove that the problem can be solved in polynomial time.

\paragraph{Graph databases, classical queries and homomorphisms.} Let $\Sigma$ be a fixed finite alphabet. A \emph{graph database} $B$ over $\Sigma$ is an arc-labeled digraph (loops and multiple arcs are allowed, unless stated otherwise). More formally, we have $B=(D_B,E_B)$ where $D_B$ is a finite digraph with vertex set $V(D_B)$ and arc set $A(D_B)$, and $E_B:A(D_B)\to\Sigma$ is an arc-label function. (As noted in the survey~\cite{F21}, in some applications, graph databases and queries may have vertex-labels, or even more complicated structures; however, such cases can in fact always be reduced to the setting of arc-labeled digraphs.)
A \emph{Boolean Conjunctive query} (CQ) $Q=(D_Q,E_Q)$ over $\Sigma$ is an arc-labeled digraph over $\Sigma$.

Given two arc-labeled digraphs $G = (D_G, E_G)$ and $H = (D_H, E_H)$, we say that there is a \emph{homomorphism} from $G$ to $H$ if there is a mapping $f$ of $V(D_G)$ to $V(D_H)$ with the property that for every arc $(u,v)$ in $A(D_G)$, there is an arc $(f(u),f(v))$ in $A(D_H)$ with $E_G(u,v)=E_H(f(u),f(v))$. If such a homomorphism exists, we also write it as $G\to H$. Moreover, we say that there is an \emph{injective homomorphism} from $G$ to $H$ if there is a homomorphism mapping $f$ such that $f$ is an injective function.

Given a graph database $B = (D_B,E_B)$ and a CQ $Q = (D_Q,E_Q)$, we say that $B$ \emph{satisfies} $Q$, denoted by $B\models Q$, if $Q\to B$. If $f$ is the homomorphism, then we say that the homomorphic image $f(Q)$ of $Q$ to $B$ is a \emph{match} of the query $Q$ in $B$.

Given two CQs $P$ and $Q$, we say that $Q$ is \emph{contained} in $P$, denoted by $Q\subseteq P$, if for every graph database $B$, if $Q\to B$ then $P\to B$.

\begin{theorem}[\cite{CH82}]
Let $P$ and $Q$ be two CQs. Then:  $Q\subseteq P$ if and only if $P\to Q$.
\end{theorem}

The most basic algorithmic problems related to queries of graph databases are the \emph{evaluation problem}, which consists in deciding whether a query $Q$ has a match in a database $B$, and the \emph{containment problem}, that is, to decide whether a query $Q$ is contained in a query $P$. These two problems can be modeled by the following decision problem.

\decisionpb{\textsc{Hom}}{Two arc-labeled digraphs $G$ and $H$.}{Does $G$ admit a homomorphism to $H$?}{0.7}

\textsc{Hom} is generally NP-complete, even when $H$ is a small fixed graph (for example a symmetric triangle, in which case \textsc{Hom} is equivalent to the graph $3$-colourability problem). To understand better the complexity of \textsc{Hom}, the following version has been studied extensively, where $H$ is a fixed arc-labeled digraph, called the \emph{non-uniform} homomorphism problem. The graph $H$ is called the \emph{template}.

\decisionpb{\textsc{Hom($H$)}}{An arc-labeled digraph $G$.}{Does $G$ admit a homomorphism to $H$?}{0.7}

As shown in~\cite{FV98}, the set of \textsc{Hom($H$)} problems captures the whole class of \emph{constraint satisfaction problems} (CSPs), whose complexity has recently been classified into polynomial cases and NP-complete cases in~\cite{B17,Z17} independently.

\paragraph{Navigational queries, navigational homomorphisms and RGPs.} In a digraph, a \emph{directed walk} is a sequence $v_1,a_{1,2},\ldots v_{k-1},a_{k-1,k},v_k$ of alternating vertices $v_i$ and arcs $a_{ij}$ of the digraph, starting and ending at vertices, such that the tail of an arc $a_{ij}$ is $v_i$ and its head is $v_j$ 
(but vertices and arcs might be repeated any number of times). A \emph{directed (simple) path} is a directed walk where each vertex occurs at most once in the sequence. In a \emph{directed trail}, each arc appears at most once, but vertices may be repeated.

Standard homomorphisms are not powerful enough to model many types of queries used in modern graph database systems. In particular, a homomorphism of a query $Q$ to a database $B$ can only match a subgraph of $B$ that is no larger than the query $Q$ itself. To the contrary, \emph{navigational queries} are types of queries where we may allow arbitrarily large subgraphs of the database to match the query. In this setting, we still model the query $Q$ (for database $B$) as an arc-labeled digraph, but the arcs are labeled with \emph{sets of words} over the alphabet $\Sigma$, rather than single letters. Now, a match of $Q$ in $B$ is a vertex-mapping $f$ from $V(D_Q)$ to $V(D_B)$ such that for an arc $(x,y)$ of $Q$ labeled with a set $E(x,y)$ of words, there exists a directed walk $W_{xy}$ in $D_B$ from $f(x)$ to $f(y)$ such that the concatenation of labels of the arcs of $W_{xy}$ is a word of $E(x,y)$.\footnote{In some applications, instead of the more general walks, trails or simple paths are considered, see for example~\cite{MP22}.}

Perhaps the most popular navigational queries are \emph{regular path queries} (RPQs), studied in many contexts~\cite{B13,BRV16,CMNP16journal,FFM08,KS08,MS04,BarceloVardi2017}. These navigational queries are based on regular expressions. An RPQ is of the form $x\xrightarrow{r}y$ where $r$ is a regular expression over the alphabet $\Sigma$. The advantage of considering such queries is that regular languages are a relatively simple yet powerful way of defining sets of words, that is both well-understood and sufficiently expressive for many applications. A \emph{conjunctive regular path query} (CRPQ) is given by $Q = e_1\land e_2\land \ldots \land e_k$ where each $e_i$ is an RPQ. Note that a CRPQ can also be viewed as an arc-labeled  digraph $Q = (D_Q,R_Q)$ where $D_Q$ is a digraph whose vertex set $V(D_Q)$ is the set of all variables occurring in $Q$ and for every RPQ $e_i = x \rightarrow{r}y$, there is an arc $a = (x,y) \in A(D_Q)$ with label $R_Q(a) = r$. Thus, CRPQs are similar to CQs, except that the arc labels are now regular expressions. To distinguish them, we follow \cite{BarceloVardi2017} and call the underlying graph structure representing a CRPQ as a \emph{regular graph pattern} (RGP). We refer to the recent survey~\cite{F21} for more details on this model.

For a fixed finite alphabet $\Sigma$, we denote by $RegExp(\Sigma)$ the set of regular expressions over alphabet $\Sigma$, with the symbols $+$ (union), $\ast$ (Kleene star), $\epsilon$ (empty string), and $\cdot$ (concatenation; sometimes this symbol is omitted). Moreover, for a regular expression $X$, as a notation we let $X^+:=X\cdot X^*$. For any regular expression $X$ in $RegExp(\Sigma)$, we denote by $L(X)$ the regular language defined by $X$.

An RGP $P$ over an alphabet $\Sigma$ is a pair $(D_P,E_P)$, where $D_P$ is a digraph with vertex set $V(D_P)$ and arc set $A(D_P)$ and $E_P:A(D_P)\to RegExp(\Sigma)$ is an arc-label function.

Given a directed walk 
$W=v_1,a_{1,2},\ldots v_{k-1},a_{k-1,k},v_k$ in an RGP $P$, where the symbols of the form $v_i$ denote the successive vertices and those of the form $a_{ij}$ denote the successive arcs, the label $E_P(W)$ of $W$ is the regular expression over $\Sigma$ formed by the concatenation $E_P(a_{1,2})\ldots E_P(a_{k-1,k})$. Note that the walk may be empty (i.e. $k=1$ and no arc is used), in which case, we let $E_P(W)=\epsilon$ (i.e. it consists only of the empty word).\footnote{This is needed in the definition, indeed, in the literature on graph databases, by convention, two vertices joined by an arc labeled with the empty word $\epsilon$ can both map to the same vertex, even if there is no loop on that vertex.}

One of the main features of evaluating regular path queries on graph databases is not only the rich expressivity of regular languages, but also, the variety of restrictions that can be put on the paths. In particular, depending on the applications, one can consider {\it directed walks} (the usual case),  {\it directed trails} or {\it directed simple paths}. The two latter variants are called the \emph{trail semantics} and \emph{simple path semantics}, respectively, and have all been considered in the literature~\cite{MW95}. When we consider CRPQs of RGPs in the trail semantics or simple path semantics, there is an additional layer to the variant: whether we allow the repetition of arcs/vertices across different paths. When this is not allowed, this corresponds to variants of injective homomorphisms, and was considered in~\cite{FR23} for the case of simple paths and repetition of vertices. Here, we also define and consider the case of repetition of arcs in the trail semantics. We next define all these variants formally.

\begin{definition}[Navigational homomorphism]
\label{def-navigational-homomorphism}
Given two RGPs $P$ and $Q$ over alphabet $\Sigma$, we define the following:
\begin{itemize}
    \item A \emph{navigational homomorphism} (\emph{n-homomorphism} for short) of $P$ to $Q$ is a mapping $f$ of $V(D_P)$ to $V(D_Q)$ such that for each arc $a=(x,y)$ in $D_P$, there is a \emph{directed walk} $W_a$ in $Q$ from $f(x)$ to $f(y)$ such that the language $L(E_Q(W_a))$ is contained in the language $L(E_P(x,y))$. When such an n-homomorphism exists, we write $P\nto Q$ and $f$ is the \emph{witness} for $P\xrightarrow{n}Q$.
    
  \item  A \emph{simple path navigational homomorphism} (\emph{spn-homomorphism} for short) of $P$ to $Q$ is a mapping $f$ of $V(D_P)$ to $V(D_Q)$ such that for each arc $a=(x,y)$ in $D_P$, if $x\neq y$, there is a \emph{simple path} $W_a$ in $Q$ from $f(x)$ to $f(y)$ such that the language $L(E_Q(W_a))$ is contained in the language $L(E_P(x,y))$. If $x=y$ ($a$ is a loop), then $W_a$ is a walk where only the first and last vertex are equal, and all other vertices are distinct. When such an spn-homomorphism exists, we write $P\xrightarrow{spn}Q$ and $f$ is the \emph{witness} for $P\xrightarrow{spn}Q$.

    Furthermore, we say that there is an \emph{injective simple path navigational homomorphism} (\emph{ispn-homomorphism} for short) of $P$ to $Q$ if $P\xrightarrow{spn}Q$ with the witness function $f$ such that (i) $f$ is injective, i.e. for any two distinct vertices $x,y$ in $V(D_P)$, $f(x)\neq f(y)$, (ii) for any two arcs $a,b$ in $A(D_P)$, the paths $W_a$ and $W_b$ do not have any internal vertices in common, and (iii) for any vertex $x\in V(D_P)$ and arc $a\in A(D_P)$, $f(x)$ is not an internal vertex of $W_a$.\footnote{In other words, the image of $D_P$ in $D_Q$ by $f$, together with the witnessing simple paths, forms a subdivision of $D_P$.} When such an ispn-homomorphism exists, we write $P\xrightarrow{ispn}Q$ and $f$ is the \emph{witness} for $P\xrightarrow{ispn}Q$.

    \item  A \emph{trail navigational homomorphism} (\emph{tn-homomorphism} for short) of $P$ to $Q$ is a mapping $f$ of $V(D_P)$ to $V(D_Q)$ such that for each arc $a = (x,y)$ in $D_P$, there is a \emph{directed trail} $W_a$ in $Q$ from $f(x)$ to $f(y)$ such that the language $L(E_Q(W_a))$ is contained in the language $L(E_P(x,y))$. When such a tn-homomorphism exists, we write $P\xrightarrow{tn}Q$ and $f$ is the \emph{witness} for $P\xrightarrow{tn}Q$.

    Furthermore, we say that there is an \emph{injective trail navigational homomorphism} (\emph{itn-homomorphism} for short) of $P$ to $Q$ if $P\xrightarrow{tn}Q$ with the witness function $f$ such that for any two arcs $a,b$, the trails $W_a$ and $W_b$ do not have any arcs in common (but they may have vertices in common). When such an itn-homomorphism exists, we write $P\xrightarrow{itn}Q$ and $f$ is the \emph{witness} for $P\xrightarrow{itn}Q$.
\end{itemize}

\end{definition}

The notion of n-homomorphism (in a more restricted form) was studied in~\cite{BRV16,BarceloVardi2017}, see also~\cite{R18}. In these works, the homomorphism maps arcs to arcs (not to paths), which provides a much more restricted concept. The notion of mapping graph arcs to paths in another graph has also been called \emph{embedding} in some works~\cite{FGKMNT20}, a term that originated in the context of interconnection networks~\cite{DPS02}. Notice also that $P\xrightarrow{ispn}Q$ implies the existence of a \emph{subdivision} of $D_P$ in $D_Q$, an important notion in graph theory.

Note that these definitions also apply to graph databases since a graph database is an RGP whose labels are all regular expressions consisting of a unique symbol of $\Sigma$. We may thus define the associated decision problem, that corresponds to the task of CRPQ \emph{evaluation}, as follows.

\decisionpb{\textsc{\rgphom}}{Two RGPs $P$ and $Q$.}{Does $P$ admit an n-homomorphism to $Q$?}{0.7}

Similarly, we have the variants \textsc{\rgphom[spn]}, \textsc{\rgphom[ispn]}, \textsc{\rgphom[tn]} and \textsc{\rgphom[itn]}, where the ``$n$-homomorphism''  in \textsc{\rgphom} is replaced by the corresponding spn-homomorphism, ispn-homomorphism, tn-homomorphism and itn-homomorphism, respectively.

Again, we may also study the non-uniform version of \textsc{\rgphom}, defined as follows for a fixed RGP $Q$. It was introduced in~\cite{BarceloVardi2017} for the restricted case where $Q$ is a graph database.

\decisionpb{\textsc{\rgphom($Q$)}}{an RGP $P$.}{Does $P$ admit an n-homomorphism to $Q$?}{0.7}

Analogously, we have the variants \textsc{\rgphom[spn](Q)}, \textsc{\rgphom[ispn](Q)}, \textsc{\rgphom[tn](Q)} and \textsc{\rgphom[itn](Q)}. In all the considered problems, the input size is the whole input RGP(s), which includes the digraph(s) (vertices and arcs) together with the sum of lengths of their arc-labels (i.e. the number of symbols in the corresponding regular expressions).

In~\cite{BarceloVardi2017}, the authors study \textsc{\rgphom($Q$)} for the specific case where $Q$ is a fixed graph database. They make a connection to the classic homomorphism problem \textsc{Hom($H$)} and show that \textsc{\rgphom($Q$)} admits a complexity dichotomy: for a specific $Q$, \textsc{\rgphom($Q$)} is either polynomial-time or NP-complete. 
Indeed, they showed it to be equivalent to the (classical) homomorphism problems a.k.a. the \emph{Constraint Satisfaction Problems}~\cite{FV98}, whose complexity delineation follows a dichotomy based on specific algebraic properties of the template $Q$, as shown independently by Bulatov~\cite{B17} and Zhuk~\cite{Z17}.

In this paper, we initiate the study of \textsc{\rgphom($Q$)} in full generality, that is when $Q$ is not just a graph database, but any RGP. 
As we will see, for this case we cannot expect a polynomial-time/NP-complete dichotomy in the style of the result of~\cite{BarceloVardi2017}, since \textsc{\rgphom($Q$)} is in fact PSPACE-hard already for very simple cases, as it can model the problem of deciding the inclusion between regular languages. Thus, it appears that \textsc{\rgphom($Q$)} where $Q$ is a general RGP merits more investigation.

Note that for two RGPs $P$ and $Q$, if $P\nto Q$, then the query $Q$ is contained in the query $P$, but there are examples where the converse does not hold (see Figure~\ref{fig:ex-query-containment}). Thus, the problem \textsc{\rgphom} for two RGPs does not fully capture \textsc{CRPQ Containment}. 
General \textsc{CRPQ Containment} (that is, not necessarily witnessed by an n-homomorphism) is known to be EXPSPACE-complete~\cite{CDLV00,FLS98} (some special cases enjoy a lower complexity, see~\cite{FGKMNT20}, and some extensions are even more intractable~\cite{RRV17}). 

When we consider the containment problem for CRPQs, the complexity depends on whether we are using directed walks, trails or simple paths. For instance, with simple path semantics, the containment problem for CRPQs is undecidable if we allow repetition of vertices and/or arcs across different paths, but becomes PSPACE-complete if it is disallowed~\cite{FR23}. 
We prove that in our case, the homomorphism-based containment problem is PSPACE-complete for all the variants, except for the semantics which allows directed walks, which is in EXPSPACE.

When we consider directed walks, to obtain better complexity bounds, we can also restrict the class of inputs. For instance, to check whether the CRPQ $P$ is contained in $Q$, an important case is when the underlying directed graphs $D_P,D_Q$ of $P,Q$ are acyclic. Such queries are common when we consider XPath queries~\cite{FFM08} since XML documents have a tree structure. If we consider the containment problem, it remains EXPSPACE-complete even with this restriction~\cite{Figueira20}. In contrast, we show that the complexity of CRPQ homomorphism-based containment decreases to PSPACE-complete when we assume that $D_P,D_Q$ are acyclic.
Refer to Table~\ref{tab:results} for a summary of these results.

\begin{table}[htbp!]
    \centering
    \begin{tabular}{@{}lcc@{}}
        \toprule
        \textbf{Semantics} & \textbf{General CRPQ} & \textbf{Homomorphism-based}  \\
         & \textbf{Containment}& \textbf{CRPQ Containment} \\
        & & (\rgphom, this paper)\\
        \midrule
        Walks        & EXPSPACE-complete \cite{CDLV00,FLS98}   & PSPACE-hard, in EXPSPACE    \\
        Walks, acyclic RGPs & EXPSPACE-complete \cite{Figueira20} &  PSPACE-complete\\
        Trails           & open & PSPACE-complete \\
        Trails, injective & open  & PSPACE-complete   \\
        Simple paths          & undecidable \cite{FR23} & PSPACE-complete  \\
        Simple paths, injective  & PSPACE-complete \cite{FR23}  & PSPACE-complete   \\
        \bottomrule
    \end{tabular}
        \caption{ Complexities for general and homomorphism-based CRPQ containment for different path semantics.}
    \label{tab:results}
\end{table}

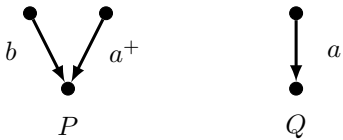
\begin{figure}[!htpb]
  \centering
  \scalebox{1}{\begin{tikzpicture}[join=bevel,inner sep=0.6mm]

      \node[draw,shape=circle,fill] (x) at (3,0) {};
      \draw (x)+(0.75,0.5) node {$a^+$};
      \draw (x)+(-0.75,0.5) node {$b$};
      \node[draw,shape=circle,fill] (y) at (2.5,1) {};
      \node[draw,shape=circle,fill] (z) at (3.5,1) {};
      \node (P) at (3,-0.5) {$Q$};
      
      \node[draw,shape=circle,fill] (u) at (0,0) {};
      \draw (u)+(0.5,0.5) node {$a$};
      \node[draw,shape=circle,fill] (v) at (0,1) {};
      \node (Q) at (0,-0.5) {$P$};
      
      \draw[->,>=latex, line width=0.4mm] (y)--(x);
      \draw[->,>=latex, line width=0.4mm] (z)--(x);
      \draw[->,>=latex, line width=0.4mm] (v)--(u);

    \end{tikzpicture}}
      \caption{Two RGPs $P$ and $Q$ over alphabet $\{a,b\}$
        which have no n-homomorphism in either direction. From $Q$ to
        $P$ because one cannot map the arc labeled by $b$,
        in the other direction because neither $b$ nor $a^+$ is
        included in $a$.
        However, any database that matches the RGP $Q$ would contain a
        walk of arcs all labeled by $a$ (because of the arc with label
        $a^+$ in $Q$). Thus, the database would clearly also match $P$. 
        So $Q$ is contained in $P$.}
  \label{fig:ex-query-containment}
\end{figure}

Another way to get better complexity bounds is to use a restricted set of regular expressions as arc labels. Such restrictions lead to better complexity bounds for many static query analysis problems of CRPQs like containment \cite{FGKMNT20} and boundedness \cite{FKM24}. Here, we consider the RGPs with labels in $\{a,a^+\}$, for some fixed symbol $a$. This setting already models not only all classic homomorphism problems and CSPs, but also some kinds of queries over hierarchical data such as SPARQL and XPath, studied for example in~\cite{CMNP16journal,FFM08,KS08,MS04} (it is known that most queries used in practice are very simple~\cite{BMT20}). With this restriction, we first show that \textsc{\rgphom} is in NP. Moreover, when we consider the injective semantics (in both trail and simple path semantics), for a fixed $Q$, the homomorphism-based containment problem is always in polynomial time. We also give a polynomial/NP-complete complexity dichotomy for \textsc{\rgphom($Q$)} in the case of undirected (or symmetric) RGPs $Q$ of this class, in the style of Hell and Ne\v{s}et\v{r}il's dichotomy for \HOM($H$)~\cite{HN90}. Furthermore, we show that even for arbitrary (directed) RGPs $Q$, \textsc{\rgphom($Q$)} follows a dichotomy, by relating it to (classical) homomorphism problems; this holds for all non-injective variants of the semantics. We then relate the case of directed path templates that have only $a$ labels to an interesting (polynomial-time solvable) parallel scheduling problem. Finally, we show that for all directed path RGP templates $Q$ with arc labels in $\{a,a^+\}$, \textsc{\rgphom($Q$)} is polynomial-time solvable. However, with injective semantics (in both trail and simple path semantics), we prove that the problem is NP-complete. All these results are summarized in Table~\ref{tab:results-with-restricted-labels}.

\begin{table}[htbp!]
    \centering
    \begin{tabular}{@{}lccc@{}}
        \toprule
        \textbf{Semantics} & \textbf{Restriction on $Q$} & \textbf{Complexity}&\textbf{Complexity when}  \\
         & & \textbf{when $Q$ is fixed}& \textbf{$Q$ is part of input} \\
        \midrule
        All        & no restriction  &in NP  & in NP  \\
         \midrule
        Non-injective semantics  & no restriction  & CSP-style P/NP-c dichotomy & NP-complete\\
         \midrule
        Undirected RGPs,  & no restriction  & Hell-Ne\v{s}et\v{r}il-style P/NP-c dichotomy & NP-complete\\
        non-injective semantics&&&\\
         \midrule
        (trail/simple paths),  & no restriction  & Polynomial time  & NP-complete\\
        injective semantics &&&\\
         \midrule
        All  & all-$a$ directed path& Polynomial time & Polynomial time\\
        \midrule
        Non-injective semantics & directed path & Polynomial time & Polynomial time\\
            \midrule
        (trail/simple paths),  & directed path& Polynomial time & NP-complete\\
         injective semantics &&&\\
        \bottomrule
    \end{tabular}
        \caption{Summary of results in this paper with arc labels restricted to $\{a,a^+\}$.}
    \label{tab:results-with-restricted-labels}
\end{table}

\paragraph{Structure of the paper.} We show in Section~\ref{sec:hardness} that \textsc{\rgphom} and related problems are PSPACE-hard (this holds even for \textsc{\rgphom($Q$)} for a simple template $Q$), by a reduction from \textsc{Regular Language Inclusion}.

We show in Section~\ref{sec:special} that \textsc{\rgphom[spn]}, \textsc{\rgphom[ispn]}, \textsc{\rgphom[tn]} and \textsc{\rgphom[itn]} are all PSPACE-complete (for any type of RGPs). Moreover, by similar arguments, for acyclic RGPs, \textsc{\rgphom} is PSPACE-complete.

We then show in Section~\ref{sec:decid} that the general case of \textsc{\rgphom} is decidable by an EXPSPACE algorithm.\footnote{Note that a wrong EXPTIME algorithm was claimed in the conference version of this paper~\cite{confversion}.} Thus, despite the additional constraint in the definition, checking homomorphism-based CRPQ containment is computationally no less efficient than checking general CRPQ containment, which is known to be EXPSPACE-complete~\cite{CDLV00,FLS98}.

Finally, in Section~\ref{sec:a-a+}, we address the simpler case of a unary alphabet $\Sigma=\{a\}$, and where all arc labels are in $\{a,a^+\}$ for some fixed symbol $a$. 
We first show that with this restriction, the complexity drops to NP (in all variants). We then prove that \textsc{\rgphom[ispn]($Q$)} and \textsc{\rgphom[itn]($Q$)} are solvable in polynomial time for every fixed RGP $Q$, but is NP-complete when $Q$ is part of the input. 
Next, we show that \textsc{\rgphom[tn]($Q$)}, \textsc{\rgphom[spn]($Q$)} and \textsc{\rgphom($Q$)} admit a P/NP-complete dichotomy, by expressing them as fixed-template CSP problems. We then show that this dichotomy can be refined for RGPs whose underlying graphs are undirected (or bidirectional and symmetric) by an explicit distinction between the polynomial and NP-complete cases, in the style of the celebrated Hell-Ne\v{s}et\v{r}il dichotomy for undirected graph homomorphisms.
When we restrict $Q$ to be a directed path, we show that the non-injective problems \textsc{\rgphom}, \textsc{\rgphom[tn]} and \textsc{\rgphom[spn]} are all polynomial-time, whereas both injective variants \textsc{\rgphom[itn]} and \textsc{\rgphom[ispn]} are NP-complete. When $Q$ is a directed path with only $a$-labels, even the latter ones become polynomial-time solvable.

We start with some preliminary considerations in Section~\ref{sec:prelim} and we conclude in Section~\ref{sec:conclu}.

\section{Preliminaries}\label{sec:prelim}

We now give some definitions and useful results from the literature.

\subsection{Regular languages}

Given a fixed alphabet $\Sigma$, a \emph{language} is a set of words over $\Sigma$. A \emph{regular expression} over $\Sigma$ is defined recursively as follows. A symbol of $\Sigma\cup\{\epsilon\}$ is a regular expression, where $\epsilon$ denotes the empty word. Given two regular expressions $E_1$ and $E_2$, $(E_1)+(E_2)$, $(E_1)\cdot (E_2)$ and $(E_1)^\ast$ are regular expressions. For a regular expression $E$, we let $L(E)$ be the language associated to $E$ in the classic way, where ``$+$'' denotes the union, ``$\cdot$'' denotes the concatenation, and ``$^\ast$'' is the Kleene star. A language  is \emph{regular} if it is the language associated to some regular expression. It is well-known that a language is regular if and only if it can be recognized by a nondeterministic finite automaton (NFA), whose number of states is proportional to the length of a regular expression representing that language. We refer to~\cite{bookREG} for further details on these matters.

We will use the following decision problems for regular languages.

  \decisionpb{\textsc{Regular Language Inclusion}}{Two regular expressions $E_1$ and $E_2$ (over the same alphabet).}{Is $L(E_1)\subseteq L(E_2)$?}{0.7}

\decisionpb{\textsc{Regular Language Universality}}{A regular expression $E$ over alphabet $\Sigma$.}{Is $L(E)=\Sigma^\ast$?}{0.7}

Note that \textsc{Regular Language Universality} is the special case of \textsc{Regular Language Inclusion} where $E_1=\Sigma^\ast$ and $E_2=E$. The following are classic results.
  
  \begin{theorem}[\cite{AHU74,MS73}]\label{thm:language-inclusion-PSPACE-c}
 \textsc{Regular Language Universality} and \textsc{Regular Language Inclusion} are PSPACE-complete.
  \end{theorem}

\subsection{Cores and n-cores}

A digraph $D$ is called a \emph{core} if it does not admit a homomorphism to a proper sub-digraph of itself; in other words, every endomorphism is an automorphism. Thus, we define similarly the notion of a \emph{navigational core}, n-core for short: an RGP $P$ is an n-core if it does not admit an n-homomorphism to a proper sub-RGP\footnote{i.e. an RGP defined by a proper subdigraph of $D_P$, and the corresponding regular expressions.} of itself. Cores and n-cores are useful because of the following fact (this is classic in the case of graphs and their cores, see the book~\cite{HNbook}). Note that we can also define the notion of $n$-core for the simple path and trail semantics. However, we do not discuss them in this paper since it is not relevant to our results.

\begin{proposition}
  Let $P$ be an RGP that is not an n-core; thus, there exists a proper sub-RGP $C$ of $P$ such that $P\nto C$. Then, for any RGP $Q$, we have $P\nto Q$ if and only if $C\nto Q$.
\end{proposition}
\begin{proof}
It is clear that $P$ and $C$ are n-homomorphically equivalent ($P\nto C$ by assumption and $C\nto P$ by the identity map, since $C$ is a sub-RGP of $P$). Thus, by transitivity of n-homomorphisms, any n-homomorphism of one of them to $Q$ directly translates into an n-homomorphism of the other to $Q$.
\end{proof}

Thus, when studying the problem \textsc{\rgphom($Q$)}, we may always assume that $Q$ is an n-core, since \textsc{\rgphom($Q$)} has the same complexity as \textsc{\rgphom($C_Q$)}, where $C_Q$ is a sub-RGP of $Q$ that is an n-core. Unfortunately, it is coNP-complete to decide whether a graph is a core~\cite{HN92} (thus deciding whether an RGP is an n-core is coNP-hard, even if it is a graph database).

Note that, with respect to classic digraph homomorphisms, any digraph $G$ has (up to isomorphism) a unique minimal subgraph to which it admits a homomorphism, called \emph{the} core of $G$. This is not the case for RGPs and n-homomorphisms. For example, any two RGPs each consisting of a unique directed cycle with all arc labels equal to $a^+$ have an n-homomorphism to each other. Thus, if we consider two such cycles of different lengths and glue them at one vertex, we obtain an RGP $P$ with two minimal sub-RGPs of $P$ (the two cycles) to which $P$ has an n-homomorphism, and these two are not isomorphic. See Figure~\ref{fig:2-cores}. A similar phenomenon has been observed in the case of a special subclass of directed tree RGPs, where being an n-core has been called \emph{non-redundancy}~\cite{CMNP16journal}.

\begin{figure}[!htpb]
  \centering
  \scalebox{1}{\begin{tikzpicture}[join=bevel,inner sep=0.6mm]

\begin{scope}
          \node[draw,shape=circle,fill] (x) at (-0.5,0) {};
      \draw (x)+(0,0.5) node {$a^+$};
      \node[draw,shape=circle,fill] (y) at (-1.5,0) {};
      \draw (y)+(0.5,-0.25) node {$a^+$};
      \node[draw,shape=circle,fill] (z) at (-1,1) {};
      \draw (z)+(-0.5,-0.4) node {$a^+$};
      \node[draw,shape=circle,fill] (t) at (0,1) {}; 
      \draw (t)+(-0.4,0.25) node {$a^+$};     
      %\node (P) at (0,-0.5) {$P$};
      \node[draw,shape=circle,fill] (u) at (1,1) {}; 
      \draw (u)+(-0.5,0.25) node {$a^+$};    
      \node[draw,shape=circle,fill] (v) at (1,0) {}; 
      \draw (v)+(0.4,0.6) node {$a^+$};     
      \draw (v)+(-0.7,0.3) node {$a^+$};        
            
      \draw[->,>=latex, line width=0.4mm] (x)--(t);
      \draw[->,>=latex, line width=0.4mm] (t)--(z);
      \draw[->,>=latex, line width=0.4mm] (z)--(y);
      \draw[->,>=latex, line width=0.4mm] (y)--(x);

      \draw[->,>=latex, line width=0.4mm] (t)--(u);
      \draw[->,>=latex, line width=0.4mm] (u)--(v);
      \draw[->,>=latex, line width=0.4mm] (v)--(t);
\end{scope}
    \end{tikzpicture}}
      \caption{An RGP with two non-isomorphic n-cores (the directed 4-cycle and the directed 3-cycle).}
  \label{fig:2-cores}
\end{figure}

For classic homomorphisms of digraphs, it is also known that the core of a digraph is always an induced subdigraph~\cite[Section 1.2]{HNbook}. This is not necessarily the case for n-cores of RGPs, since some arcs may be mapped to paths (for example, a single arc $xy$ labeled $a^+$ can be mapped to a directed path from $x$ to $y$ of arbitrary length, all whose arcs are labeled $a$). If the RGP consists only of these vertices and arcs, then the n-core is the same RGP with the $a^+$ arc deleted, which is not an induced sub-RGP. See Figure~\ref{fig:core} for an illustration.
  
\begin{figure}[!htpb]
  \centering
  \scalebox{1}{\begin{tikzpicture}[join=bevel,inner sep=0.6mm]

\begin{scope}
          \node[draw,shape=circle,fill] (x) at (0,0) {};
      \draw (x)+(0.4,0.5) node {$a^+$};
      \node[draw,shape=circle,fill] (y) at (-1,0) {};
      \draw (y)+(0.5,-0.25) node {$a$};
      \node[draw,shape=circle,fill] (z) at (-1,1) {};
      \draw (z)+(-0.25,-0.5) node {$a$};
      \node[draw,shape=circle,fill] (t) at (0,1) {}; 
      \draw (t)+(-0.5,0.25) node {$a$};     
      %\node (P) at (0,-0.5) {$P$};
            
      \draw[->,>=latex, line width=0.4mm] (t)--(x);
      \draw[->,>=latex, line width=0.4mm] (t)--(z);
      \draw[->,>=latex, line width=0.4mm] (z)--(y);
      \draw[->,>=latex, line width=0.4mm] (y)--(x);
\end{scope}
\begin{scope}[xshift=3cm]
         \node[draw,shape=circle,fill] (x) at (0,0) {};
      %\draw (x)+(0.4,0.5) node {$a^+$};
      \node[draw,shape=circle,fill] (y) at (-1,0) {};
      \draw (y)+(0.5,-0.25) node {$a$};
      \node[draw,shape=circle,fill] (z) at (-1,1) {};
      \draw (z)+(-0.25,-0.5) node {$a$};
      \node[draw,shape=circle,fill] (t) at (0,1) {}; 
      \draw (t)+(-0.5,0.25) node {$a$};     
      %\node (P) at (0,-0.5) {$P$};
            
      %\draw[->,>=latex, line width=0.4mm] (t)--(x);
      \draw[->,>=latex, line width=0.4mm] (t)--(z);
      \draw[->,>=latex, line width=0.4mm] (z)--(y);
      \draw[->,>=latex, line width=0.4mm] (y)--(x); 
\end{scope}

    \end{tikzpicture}}
      \caption{An RGP (left) and its n-core (right).}
  \label{fig:core}
\end{figure}

\section{PSPACE-hardness of \textsc{\rgphom} and related problems}\label{sec:hardness}

In this section, we show that \textsc{\rgphom} and related problems are PSPACE-hard.

\subsection{The uniform case}

We first present a very simple reduction from \textsc{Regular Language Inclusion} to \textsc{\rgphom}.

\begin{proposition}\label{prop:uniform-N-HOM-PSPACE-h}
  \textsc{\rgphom}, \textsc{\rgphom[spn]}, \textsc{\rgphom[ispn]}, \textsc{\rgphom[tn]} and \textsc{\rgphom[itn]} are all PSPACE-hard.
\end{proposition}
\begin{proof}
Note that the same reduction works for all the variants. We reduce from \textsc{Regular Language Inclusion}, which is PSPACE-complete  
(Theorem~\ref{thm:language-inclusion-PSPACE-c}). Given an input $E_1$, $E_2$ of \textsc{Regular Language Inclusion}, we construct two RGPs $P_1$ and $P_2$ on two vertices each, where $P_1$ contains a single arc labeled $XE_1$ and $P_2$, a single arc labeled $XE_2$ (where $X$ is a new symbol not belonging to the initial alphabet). Now, since $\epsilon\notin L(XE_1)$, the only possible mapping is to map the arc of $P_1$ to the arc of $P_2$, and we have $L(E_1)\subseteq L(E_2)$ if and only if $P_2\nto P_1$.  
\end{proof}

As witnessed by the simplicity of the reduction given in Proposition~\ref{prop:uniform-N-HOM-PSPACE-h}, the PSPACE-hardness of \textsc{\rgphom} and its variants is inherently caused by the hardness of the underlying regular language problem.

\subsection{The non-uniform case}

We now show that even the ``non-uniform'' (i.e. for a fixed target RGP) version of the \textsc{\rgphom} problem and its variants still remains PSPACE-hard, even for a very simple template.

\begin{proposition}\label{prop:nonuniform-N-HOM-PSPACE-c}
Let $\Sigma$ be a fixed alphabet of size at least~$2$, and let $D_2^\Sigma$ be the RGP with two vertices over $\Sigma \cup \{X\}$ (where $X\notin\Sigma$) consisting of a single arc labeled $X\Sigma^\ast$. Then, \textsc{\rgphom($D_2^\Sigma$)}, \textsc{\rgphom[spn]($D_2^\Sigma$)}, \textsc{\rgphom[ispn]($D_2^\Sigma$)}, \textsc{\rgphom[tn]($D_2^\Sigma$)} and \textsc{\rgphom[itn]($D_2^\Sigma$)} are all PSPACE-complete.
\end{proposition}
\begin{proof}
Again, similar argument works for all the variants. 

To see that the problem is in PSPACE, note that the homomorphism part of the problem is simple. Let $P$ be the input RGP. An n-homomorphism of $P$ to $D_2^\Sigma$ exists if and only if for every arc $xy$ labeled $E$ of $P$, if (i) $x$ and $y$ are mapped to the same vertex of $D_2^\Sigma$, then $\epsilon\in L(E)$, and (ii) if $xy$ is mapped to the arc in $D_2^\Sigma$, $X\Sigma^\ast\subseteq L(E)$. Condition (i) can be checked easily in polynomial-time, and for (ii), we need to solve \textsc{Regular Language Universality} for each arc-label $E$. Since this problem is in PSPACE (\cite{AHU74}, see Theorem~\ref{thm:language-inclusion-PSPACE-c}), \textsc{\rgphom($D_2^\Sigma$)} is in PSPACE as well, by enumerating all possible mappings of $D_P$ to $D_2^\Sigma$.
  
  To show that \textsc{\rgphom($D_2^\Sigma$)} is PSPACE-hard, we reduce \textsc{Regular Language Universality} to it. Given an input $E$ (over alphabet $\Sigma$) of \textsc{Regular Language Universality}, we construct the RGP $P$ on two vertices consisting of a single arc labeled $XE$. Now, since $\epsilon\notin L(XE)$, we have $X\Sigma^\ast\subseteq L(XE)$ if and only if $P\nto D_2^\Sigma$.
\end{proof}

\subsection{Testing for being an n-core}

We now give a similar reduction that shows that testing if an RGP is an n-core is also PSPACE-hard. We define the following decision problem:

\decisionpb{\textsc{N-Core}}{an RGP $P$.}{Is $P$ an n-core?}{0.7}

\begin{proposition}\label{prop:N-core-PSPACE}
\textsc{N-Core} is PSPACE-hard.
\end{proposition}
\begin{proof}
We reduce \textsc{Regular Language Inclusion} to \textsc{N-Core}. Given an input $E_1$, $E_2$ (where $\epsilon \not\in L(E_1),L(E_2)$) of \textsc{Regular Language Inclusion} (over alphabet $\Sigma$), we construct an RGP $P(E_1,E_2)$ over alphabet $\Sigma\cup\{X\}$, where $X\notin\Sigma$. We have $V(P)=\{x,y,z\}$ and $P$ contains an arc $(x,y)$ labeled $E_1$ and an arc $(x,z)$ labeled $X+E_2$. 
Now, there are only two potential proper subgraphs to consider as candidates for an n-core, one for each arc. In both cases, the putative witness n-homomorphism would leave $x$ fixed. The case mapping $y$ to $z$ is not possible as $L(X+E_2) \subseteq  L(E_1)$ cannot hold since $X$ is not in $L(E_1)$ by assumption. The case mapping $z$ to $y$ would require that $L(E_1) \subseteq  L(X+E_2)=X\cup L(E_2)$ which is equivalent to  $L(E_1) \subseteq  L(E_2)$ since $X$ is not in the alphabet of $E_1$.

Hence, we have $L(E_1)\subseteq L(E_2)$ if and only if $P$ is not an n-core. Since PSPACE=coPSPACE, we are done. 
\end{proof}

\section{Special cases that are in PSPACE}\label{sec:special}

We first show that all the variants: \textsc{\rgphom[spn]}, \textsc{\rgphom[ispn]}, \textsc{\rgphom[tn]} and \textsc{\rgphom[itn]} are PSPACE-complete.

\begin{theorem}
\label{thm-variants-pspace}
 The problems \textsc{\rgphom[spn]}, \textsc{\rgphom[ispn]}, \textsc{\rgphom[tn]} and \textsc{\rgphom[itn]} are all PSPACE-complete.
\end{theorem}
\begin{proof}
    The lower bound follows from Proposition~\ref{prop:uniform-N-HOM-PSPACE-h}. 
    
    For the upper bound, first note that for all the variants we cannot repeat the arcs and/or vertices. Thus, for every arc $(x,y)$ in $A(D_P)$ and for the corresponding witness mapping $f$, the path from $f(x)$ to $f(y)$ will always be of polynomial length in the size of $|Q|$. Thus, the (non-deterministic) PSPACE algorithm can guess the witness mapping $f$ and in the \textsc{\rgphom[spn]} and \textsc{\rgphom[tn]} cases, for every arc $a=(x,y)$ in $A(D_P)$, it can guess the path/trail $W_a$ (of polynomial size) and verify that $L(E_Q(W_a))$ is contained in the language $L(E_P(x,y))$.

    For the \textsc{\rgphom[ispn]} and \textsc{\rgphom[itn]} cases, if $A(D_P)= \{a_1,\ldots a_k\}$ then the algorithm can first guess the witness mapping $f$, and all the paths $W_{a_1},\ldots, W_{a_k}$ and verify that there are no repeated arcs/vertices across any of the paths and finally check that for every $i\le k$, $L(E_Q(W_{a_i}))$ is contained in the language $L(E_P(a_i))$.

    All guesses are of polynomial size and all the checks can be done in PSPACE (since we need to check for regular language inclusion). Hence we have a non-deterministic PSPACE algorithm and since NPSPACE=PSPACE, we obtain the PSPACE upper bound.
\end{proof}

We obtain the following consequence of Proposition~\ref{prop:nonuniform-N-HOM-PSPACE-c} and Theorem~\ref{thm-variants-pspace}.

\begin{corollary}
    There exists an RGP $Q$ such that the problems \textsc{\rgphom[spn]($Q$)}, \textsc{\rgphom[ispn]($Q$)}, \textsc{\rgphom[tn]($Q$)} and \textsc{\rgphom[itn]($Q$)} are all PSPACE-complete.
\end{corollary}

Now we turn to \textsc{\rgphom}. In this case, a sufficient condition to obtain a PSPACE upper bound is when the RGP $Q$ is acyclic.

\begin{theorem}
   \textsc{\rgphom} is in PSPACE if the RGP $Q$ has an underlying directed graph $D_Q$ that is acyclic. Moreover, there exists some RGP $Q$ such that \textsc{\rgphom($Q$)} is PSPACE-complete.
\end{theorem}
\begin{proof}
    The lower bound follows from Proposition~\ref{prop:nonuniform-N-HOM-PSPACE-c}. For the upper bound, note that every path in $Q$ can have at most length $|A(D_Q)|$. Hence, similar to Theorem~\ref{thm-variants-pspace}, for every arc $a=(x,y)$ in $A(D_P)$, one can guess the path/trail $W_a$ (of polynomial size) and verify that $L(E_Q(W_a))$ is contained in the language $L(E_P(x,y))$.
\end{proof}

\section{An EXPSPACE upper bound for \textsc{\rgphom}}\label{sec:decid}

In this section, we show that \textsc{\rgphom} is decidable by an EXPSPACE algorithm.

\begin{lemma}
\label{lemma:path-length-bound}
Let $P = (D_{P},E_{P})$ be an RGP and $r$ be a regular expression. For every $x,y\in V(D_{P})$, there is a walk $W$ from $x$ to $y$ with $L(E_{W}(x,y)) \subseteq L(r)$ if and only if there is a walk $W'$ from $x$ to $y$ with $L(E_{W'}(x,y)) \subseteq L(r)$ in which every arc $e\in A(D_{P})$ appears at most $2^{m}$ times, where $m$ is the number of states of a smallest DFA for $L(r)$.
\end{lemma}
\begin{proof}
It suffices to prove $(\Rightarrow)$, since $(\Leftarrow)$ is trivial. Let $\mathcal{A} = (\mathcal{Q},\Sigma,q_0,\delta,F)$ be a minimal DFA for $L(r)$, where $|\mathcal{Q}| = m$.
Let $W = u_0e_0u_1e_1u_2\ldots u_{k-1}e_{k-1}u_k$ be the shortest walk such that $u_0 = x$, $u_k = y$ and $L(E_{W}(x,y)) \subseteq L(r)$.
Let $L_i$ denote the language associated with the label of arc $e_i$ i.e., $L_i = L(E(e_i))$. For brevity, we denote $L(e_0\ldots e_i)$ to be the language corresponding to $L(E(e_0)\cdots E(e_i))$.

If every arc $e\in A(D_{P})$ appears at most $2^m$ times in $W$, we are done. So, we next assume that there is some arc $e$ that occurs more than $2^m$ times in $W$.

Note that, since $\mathcal{A}$ is a DFA, for every word $w$ and state $q$, the result of $\delta(q,w)$ is a single state of $A$ which denotes the state at which the transition ends on the word $w$ when started from the state $q$. For every $i\le k$, define $T_i = \bigcup\limits_{w\in L(e_0\cdots e_i)}\delta(q_0,w)$, that is, the set of states reachable from $q_0$ by some word in the language of the concatenation of the labels of the first $i$ arcs of $W$. Note that, for every $q\in T_i$ and for every word $v\in L(e_{i+1}\ldots e_{k-1})$ we have $\delta(q,v)\in F$. Indeed, by definition of $T_i$, there is a word $w_q$ in $L(e_0\ldots e_i)$ with $\delta(q_0,w_q)=q$. Then, the word $w_qv$ belongs to $L(e_0\ldots e_{k-1})=L(E_{W}(x,y))$, and since $L(E_{W}(x,y)) \subseteq L(r)$, $w_qv$ is accepted by $\mathcal{A}$ and thus $\delta(q,v)=\delta(q_0,w_qv)\in F$.

\medskip
Now, since $e$ appears more than $2^m$ times in $W$ and there are $2^m$ possible subsets of states of $\mathcal{A}$, there exists $i < j$ such that $e_i = e_j = e$ (and hence $u_{i} = u_{j}$ and $u_{i+1} = u_{j+1}$), and moreover, $T_i = T_j$. Let $u_i = u_j = a$ and $u_{i+1} = u_{j+1} = b$. With this notation, we can rewrite $W$ as $u_0e_0u_1e_1u_2\ldots e_{i-1}a{e}b\ldots aebe_{j+1}\ldots u_{k-1}e_{k-1}u_k$.

Consider the walk $W'$ obtained by removing from $W$ the arcs between $e_{i}$ and $e_{j+1}$, that is, $W' = u_0e_0u_1e_1u_2\ldots e_{i-1}a{e}be_{j+1}\ldots u_{k-1}e_{k-1}u_k$. Clearly, $W'$ is a walk shorter than $W$ (since we removed at least one occurrence of $e$ from $W$, and possibly, some additional arcs) that starts at $x$ and ends at $y$. We will prove that $L(E_{W'}(x,y)) \subseteq L(r)$, which contradicts the minimality of $W$.

Towards this, consider $w\in L(E_{W'}(x,y))$. Let $w = w_1w_2$ where $w_1\in L(e_0\ldots e_{i-1}e)$ and $w_2 \in L(e_{i+1}\ldots e_{k-1})$. Let $\delta(q_0,w_1) = q\in T_i$. Since $T_i = T_j$, there exists a word $w_1'\in L(e_0\ldots e_j)$ such that $\delta(q_0,w_1') = q$. Moreover, $w' = w_1'w_2 \in L(e_0\ldots e_je_{j+1}\ldots e_{k-1})=L(E_{W}(x,y))$ and hence $w'$ also belongs to $L(r)$, so $\delta(q,w_2) = f \in F$. Hence, by composition, and since $\mathcal{A}$ is a DFA, $\delta(q_0,w_1w_2) = f\in F$.
\end{proof}

\begin{lemma}
\label{lemma:small-string-test}
Let $\mathcal{A}=(\mathcal{Q},q_0,\delta,F)$ be a DFA and $r$ be a regular expression and $T,T' \subseteq \mathcal{Q}$. Consider Property $\Pi$: for every $w\in L(r)$ and $q\in T$, $\delta(q,w)\in T'$. 
Checking whether $\Pi$ is satisfied or not is polynomial-time solvable (in the size of $\mathcal{A}$, $r$, $T$ and $T'$).
\end{lemma}
\begin{proof}

For every $q\in \mathcal{Q}$ define:
$$\mathsf{reach}(q) = \{ s\mid s\in \mathcal{Q} \text{ and there exists } w\in L(r) \text{ such that the run of } \mathcal{A} \text{ on } w \text{ starting from } q \text{ ends in } s\}.$$ 

We claim that for every $q\in \mathcal{Q}$, $\mathsf{reach}(q)$ is computable in polynomial time (in the size of $\mathcal{A}$ and $r$). Assuming the claim, the lemma follows because checking the property $\Pi$ is the same as checking that, for every $q\in T$, $\mathsf{reach}(q) \subseteq T'$.

Thus, it remains to prove the claim. For this, let $\mathcal N_r = (\mathcal{Q}_r,q_r,\delta_r,F_r)$ be an NFA accepting $L(r)$. Let $n_r = |\mathcal{Q}_r|$. Note that $n_r$ is linear in the size of $r$.
 Further, construct $\mathcal{B} = \mathcal{A}\times \mathcal{N}_r$ which is the standard cross-product of the two automata. Note that the size of $\mathcal{B}$ is polynomial in the size of $\mathcal A$ and $r$.

\medskip

 Now, for any $q\in \mathcal{Q}$ we first verify the following:
 $$\mathsf{reach}(q) = \{ s\mid (s,q_f) \text{ is a state in } \mathcal{B} \text{ such that } q_f \in F_r \text{ and there is a path in } \mathcal{B} \text{ from } (q,q_r) \text{ to } (s,q_f)\}.$$

 Towards the proof, pick some $s\in \mathsf{reach}(q)$. By definition, there is some $w\in L(r)$ such that its run in $\mathcal{A}$ starts at $q$ and ends in $s$ and there is a run of $w$ in $\mathcal{N}_r$ starting from $q_r$ and reaching some $q_f\in F_r$. Thus, there is a path in $\mathcal{B}$ from $(q,q_r)$ to $(s,q_f)$.

 Conversely, suppose $(s,q_f)$ is a state in $\mathcal{B}$ such that $q_f \in F_r$ and there is a path in $\mathcal{B}$ from $(q,q_r)$ to $(s,q_f)$. Let $w$ be the word along this path. The run of $w$ in $\mathcal{N}_r$ starts at $q_r$ and ends in $q_f\in F$, which implies $w\in L(r)$. Moreover, the run of $w$ in $\mathcal{A}$ starts at $q$ and ends in $s$. Thus, by definition, $s\in \mathsf{reach}(q)$.

 \medskip
 Thus, computing $\mathsf{reach}(q)$ is reduced to the reachability problem in $\mathcal{B}$, which can be done in polynomial time (since the size of $\mathcal{B}$ is polynomial in the size of $\mathcal{A}$ and $r$).
\end{proof}

\begin{theorem}
\label{thm-rgphom-in-expspace}
\rgphom\ is in EXPSPACE.
\end{theorem}
\begin{proof}
We give a non-deterministic EXPSPACE algorithm (note that EXPSPACE=NEXPSPACE~\cite{savitch}).

Let two RGPs $P = (D_P,E_P)$ and $Q = (D_Q,E_Q)$ be the instance of \textsc{\rgphom}. We perform the following steps.

\begin{enumerate}
    \item Guess a mapping $h : D_P \to D_Q$ that would witness that the instance is a YES-instance.
    \item Consider each arc $a\xrightarrow{r}b \in E_P$, and denote 
    $h(a) = u$ and $h(b) = v$. 
    \begin{enumerate}
    \item Construct $\mathcal{A} = (\mathcal{Q},q_0,\delta,F)$, a smallest DFA for $L(r)$.
        \item Guess the walk $W$ from $u$ to $v$ in $D_Q$ that would witness $L(E_Q(W))\subseteq L(r)$.
        \item Check that indeed, $L(E_Q(W))\subseteq L(r)$.
    \end{enumerate}
\end{enumerate}

Every mapping $h$ can be expressed in polynomial space, and mappings can be enumerated easily. 
Step~2 has polynomially many iterations, since there are polynomially many arcs in $W$. Next, we show that all these iterations can be done in EXPSPACE.

Consider an arc $a\xrightarrow{r}b \in A(D_P)$ with $h(a) = u$ and $h(b) = v$.

Constructing, in Step~(a), $\mathcal{A} = (\mathcal{Q},q_0,\delta,F)$, a minimal DFA for $L(r)$, can be done in EXPTIME, and $\mathcal{A}$ can be stored in exponential space. Indeed, if $m$ is the number of states of $\mathcal{A}$, we have $m \le 2^{|r|}$.

We now show that steps $(b)$ and $(c)$ can be verified together in EXPSPACE. Since we are giving a non-deterministic algorithm, assume that there exists a walk $W$ from $u$ to $v$ in $D_Q$ such that $L(E_Q(W))\subseteq L(r)$. From Lemma \ref{lemma:path-length-bound}, we can assume that the length of $W$ is at most $|E_Q|2^m = O( 2^{2^{|r|}})$. Guessing the entire $W$ in one shot will take double exponential space. Instead, we will guess $W$, one arc at a time. We will also maintain a counter {\tt count} initialized to $|E_Q|2^m$ (which requires a singly exponential space) and at every step when a new arc is guessed, we will decrement {\tt count}. Since the input is a YES-instance, we will guess the entire $W$ before {\tt count} value becomes {\tt 0}.

At every step $i$ let $e_0e_1\ldots e_i$ be the partial walk guessed.  Note that we do not remember the entire walk up to step $i$. Instead, we maintain a node $w \in V(D_Q)$ and a set $T\subseteq \mathcal{Q}$ with the intended meaning that $w$ is the last vertex in $W$ that has been guessed until step $i$ and $T$ is the set of states in $\mathcal{Q}$ such that for every word $s\in L(E_Q(e_0e_1\ldots e_i))$, we have $\delta(q_0,s)\in T$. 

At step $0$, $w = u$ and $T=\{q_0\}$. Suppose we have constructed $W$ up to step $i$. If $w = v$ and $T\subseteq F$, then, from the invariant that we have maintained, the algorithm has indeed found a required path. Else, if ${\tt count} =0$, then the algorithm returns {\tt failure}. Otherwise, we decrement {\tt count} and guess an outgoing arc from $w$, say $w\xrightarrow{e}w'$, which will be the next step in the path $W$. We will also guess $T' \subseteq \mathcal{Q}$ which will satisfy the required invariant.

It remains to prove that we can verify in EXPSPACE that $T'$ satisfies the invariant. By Lemma~\ref{lemma:small-string-test}, we can verify in polynomial time that for every $s'\in L(E_Q(e))$ and $q\in T$, we have $\delta(q,s')\in T'$. But this is sufficient to verify the invariant. Pick any arbitrary word $s\in e_0e_1\ldots e_ie$. We can write $s = s's''$ where $s'\in L(e_0e_1\ldots e_i)$ and $s''\in L(e)$. From the invariant, we have $\delta(q_0,s') = q\in T$ and we have verified that $\delta(q,s'') = q'\in T'$, which is the required invariant.

Hence, the whole of (b) and (c) can be done in EXPSPACE.

The correctness of the algorithm follows from the fact that we iterate over all possible homomorphisms, and if we have a YES-instance, then the witnessing walks are all of bounded length by Lemma~\ref{lemma:path-length-bound}. Hence, the algorithm returns YES if and only if the sought homomorphism exists.
\end{proof}

We obtain the following corollary.

\begin{corollary}\label{cor:n-core}
\textsc{N-Core} is in EXPSPACE.
\end{corollary}
\begin{proof}
To decide whether a given RGP $P$ is an n-core, it suffices to check, for each proper sub-RGP $Q$ of $P$, whether we have $P\nto Q$. Thus, an exponential number of applications of our EXPSPACE algorithm for \textsc{\rgphom} is a valid EXPSPACE algorithm for \textsc{N-Core}.
\end{proof}

\section{Specific RGP classes over a unary alphabet: the $\{a,a^+\}$ case}\label{sec:a-a+}

In this section, we consider that the alphabet $\Sigma$ is unary, say, $\Sigma=\{a\}$. It is known that some unary languages are undecidable (because the set of unary languages is uncountable, while the set of decidable languages is countable). Thus, even unary languages are highly nontrivial and it is of interest to study them. Note that for unary \emph{regular} languages, \textsc{Regular Language Inclusion} and \textsc{Regular Language Universality} are no longer PSPACE-complete but they are coNP-complete (see~\cite{HRS76} and~\cite{MS73}, respectively).

The case where all arc-labels of the considered RGPs are equal to $a$ is equivalent to the problem of classic digraph homomorphisms, and thus it captures all CSPs, see~\cite{FV98}. When all arc-labels are in $\{a,a^+\}$, we have two kinds of constraints: arcs labeled $a$ must map in a classic, local, way, while arcs labeled $a^+$ can be mapped to an arbitrary (nontrivial) path in the target RGP. Thus, this setting is useful for example to model descendence relations in hierarchical data such as XML. This setting is for example used in languages like SPARQL or XPath for XML documents, that are tree-structured. We refer for example to the papers~\cite{CMNP16journal,FFM08,KS08,MNS09,MS04}. For digraphs, this type of homomorphisms is related to the setting where one considers the \emph{graph power} of the template digraph, a problem studied in~\cite{powers}.

\begin{observation}\label{obs:a-equiv-noninjective}
In the setting of arc-labels in $\{a,a^+\}$, \textsc{\rgphom[tn]($Q$)}, \textsc{\rgphom[spn]($Q$)} and \textsc{\rgphom($Q$)} are equivalent, since the existence of a walk implies the existence of a trail which further implies existence of a simple path.
\end{observation}

We first show how to transform the non-injective variants in this setting into a classic homomorphism problem. For an RGP $Q$ with arc-labels in $\{a,a^+\}$, let $D(Q)$ be the two-arc-labeled digraph obtained from $Q$ by leaving all arcs labeled $a$ untouched, and adding an arc with label $t$ from a vertex $x$ to a vertex $y$ if and only if there is a directed path (regardless of any labels) from $x$ to $y$ in $Q$ (that is, the arcs labeled $t$ induce the transitive closure of the underlying digraph $D_Q$). This construction is computable in polynomial time. See Figure~\ref{fig:path-D(Q)} for an example (we adopt the style of~\cite{CMNP16journal} by marking the $a^+$ arcs as doubled).

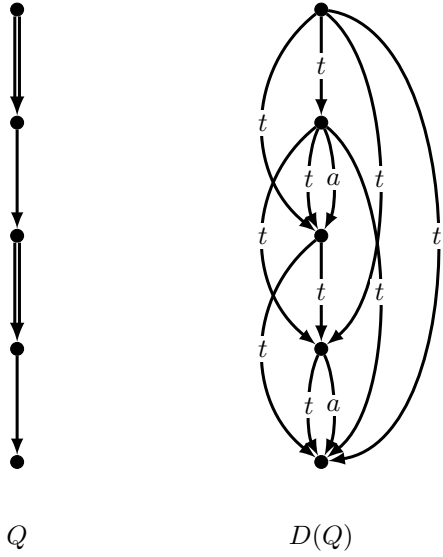
\begin{figure}[!htpb]
  \centering
  \scalebox{1}{\begin{tikzpicture}[join=bevel,inner sep=0.6mm]
      \begin{scope}
        \node[draw,shape=circle,fill] (a) at (0,0) {};
        \node[draw,shape=circle,fill] (b) at (0,1.5) {};
        \node[draw,shape=circle,fill] (c) at (0,3) {};
        \node[draw,shape=circle,fill] (d) at (0,4.5) {};
        \node[draw,shape=circle,fill] (e) at (0,6) {};
        
        \node (Q) at (0,-1) {$Q$};
        
        \draw[->,>=latex,double,line width=0.4mm] (e)--(d);
        \draw[->,>=latex,line width=0.4mm] (d)--(c);
        \draw[->,>=latex,double,line width=0.4mm] (c)--(b);
        \draw[->,>=latex,line width=0.4mm] (b)--(a);
      \end{scope}

      \begin{scope}[xshift=4cm]
        \node[draw,shape=circle,fill] (a) at (0,0) {};
        \node[draw,shape=circle,fill] (b) at (0,1.5) {};
        \node[draw,shape=circle,fill] (c) at (0,3) {};
        \node[draw,shape=circle,fill] (d) at (0,4.5) {};
        \node[draw,shape=circle,fill] (e) at (0,6) {};
        \node (DQ) at (0,-1) {$D(Q$)};

        \draw[->,>=latex,line width=0.4mm] (e)--(d) node[midway,fill=white] () {$t$};
        \draw[->,>=latex,line width=0.4mm] (d) .. controls ++(0.2,-0.4) and ++(0.2,0.4) .. (c) node[midway,fill=white] () {$a$};
        \draw[->,>=latex,line width=0.4mm] (d) .. controls ++(-0.2,-0.4) and ++(-0.2,0.4) .. (c) node[midway,fill=white] () {$t$};
        \draw[->,>=latex,line width=0.4mm] (c)--(b) node[midway,fill=white] () {$t$};
        \draw[->,>=latex,line width=0.4mm] (b) .. controls ++(0.2,-0.4) and ++(0.2,0.4) .. (a) node[midway,fill=white] () {$a$};
        \draw[->,>=latex,line width=0.4mm] (b) .. controls ++(-0.2,-0.4) and ++(-0.2,0.4) .. (a) node[midway,fill=white] () {$t$};
        \draw[->,>=latex,line width=0.4mm] (e) .. controls ++(-1,-0.8) and ++(-1,0.8) .. (c) node[midway,fill=white] () {$t$};
        \draw[->,>=latex,line width=0.4mm] (d) .. controls ++(-1,-0.8) and ++(-1,0.8) .. (b) node[midway,fill=white] () {$t$};
        \draw[->,>=latex,line width=0.4mm] (c) .. controls ++(-1,-0.8) and ++(-1,0.8) .. (a) node[midway,fill=white] () {$t$};
        \draw[->,>=latex,line width=0.4mm] (e) .. controls ++(1,-0.8) and ++(1,0.8) .. (b) node[midway,fill=white] () {$t$};
        \draw[->,>=latex,line width=0.4mm] (d) .. controls ++(1,-0.8) and ++(1,0.8) .. (a) node[midway,fill=white] () {$t$};
        \draw[->,>=latex,line width=0.4mm] (e) .. controls ++(2,-0.3) and ++(2,0.3) .. (a) node[midway,fill=white] () {$t$};
      \end{scope}

  \end{tikzpicture}}
  \caption{A path RGP $Q$ with arc-labels in $\{a,a^+\}$ (doubled arcs are labeled $a^+$, the others are labeled $a$) and the corresponding two-arc-labeled digraph $D(Q)$.}
  \label{fig:path-D(Q)}
\end{figure}

\begin{proposition}\label{prop:a-a+-reduction-to-HOM}
\textsc{\rgphom[tn]}, \textsc{\rgphom[spn]} and \textsc{\rgphom} restricted to input RGPs $P$ and $Q$ with arc-labels in $\{a,a^+\}$ are polynomially reducible to \textsc{\HOM} for some 2-arc-labeled input digraphs $D'_P$ and $D(Q)$, where the latter is defined as above.
\end{proposition}
\begin{proof}
By Observation~\ref{obs:a-equiv-noninjective}, it suffices to show the statement for \textsc{\rgphom}.

Let $D_P'$ be the 2-arc-labeled digraph obtained from $D_P$ by replacing all arc-labels $a^+$ by labels $t$, and keeping the arcs labeled $a$. Note that this is a different construction from $D(P)$.

Now, for any instance $P,Q$ of \textsc{\rgphom}, we claim that $P\nto Q$ (as a navigational homomorphism) if and only if $D'_P\to D(Q)$ (as a classic homomorphism of arc-labeled digraphs). If $P\nto Q$, let $f$ be a homomorphism witnessing it. We show that $f$ also is a homomorphism of $D'_P$ to $D(Q)$. Let $(x,y)$ be an arc in $D'_P$. If $(x,y)$ is labeled $a$, it means $(x,y)$ was also labeled $a$ in $P$, and thus $(f(x),f(y))$ is labeled $a$ in both $Q$ and $D(Q)$. If $(x,y)$ is labeled $t$, then $(x,y)$ was labeled $a^+$ in $P$. Thus, there is a directed path from $f(x)$ to $f(y)$ in $Q$; thus, $(f(x),f(y))$ is labeled $t$ in $D(Q)$. Hence, indeed $f$ witnesses that $D'_P\to D(Q)$. The proof is similar in the other direction.
\end{proof}

We show next that the special case we consider here is much easier than the general PSPACE-hard case.

\begin{proposition}\label{prop:a-a+-NP}
If the input RGPs are restricted to those which have arc-labels in $\{a,a^+\}$ where $a\in \Sigma$, then \textsc{\rgphom[tn]}, \textsc{\rgphom[spn]}, \textsc{\rgphom}, \textsc{\rgphom[itn]} and \textsc{\rgphom[ispn]} are in $NP$.
\end{proposition}
\begin{proof}
For \textsc{\rgphom[tn]}, \textsc{\rgphom[spn]} and \textsc{\rgphom}, it suffices to apply Proposition~\ref{prop:a-a+-reduction-to-HOM}, by noting that \HOM{} is in NP.

For \textsc{\rgphom[itn]} and \textsc{\rgphom[ispn]}, we can use a direct argument; indeed, a YES-instance can be witnessed by the mapping $f$ of $D_P$ to $D_Q$, and for each arc $(x,y)$ of $A(D_P)$, a simple path (or trail) from $f(x)$ to $f(y)$ in $D_Q$. Moreover, these simple paths or trails must respect the injectivity conditions of the respective semantics (vertex-disjoint or edge-disjoint). If the arc $(x,y)$ is labeled $a$, the simple path or trail must actually be a single arc labeled $a$. Since this witness has polynomial size in $Q$, these problems are in NP.
\end{proof}

\subsection{Injective semantics}

We first show that in the injective semantics, the non-uniform variants of \textsc{\rgphom} (i.e. when $Q$ is fixed) are polynomial-time solvable.

\begin{theorem}
    For any fixed RGP $Q$ over alphabet $\Sigma$ with arc-labels in $\{a,a^+\}$, where $a\in\Sigma$, \textsc{\rgphom[itn]($Q$)} and \textsc{\rgphom[ispn]($Q$)} are polynomial-time solvable for input RGPs with arc-labels in $\{a,a^+\}$.
\end{theorem}
\begin{proof}
    Let the number of vertices in $D_Q$ be $k$ and the number of arcs be $\ell$. Given an input RGP $P$, let $D_P$ contain $n$ vertices and $m$ arcs. Let the set of arcs in $D_P$ be $\{e_1,\ldots e_m\}$ with $x_i$ and $y_i$ being the source and target of the arc $e_i$ respectively.

\medskip

    First, consider \textsc{\rgphom[itn]($Q$)}. Since we are looking for an injective mapping, every arc of $D_Q$ can be part of the witness path of at most one arc in $D_P$. Hence, if $D_P$ contains more arcs than $D_Q$ ($m>\ell$) then by the pigeonhole principle, there has to be some arc in $D_Q$ which is part of the witness path of two different arcs of $D_P$, a contradiction. In this case, we can return NO.
    So, we can assume that $m\le \ell$. 

    Now, we can simply map all isolated vertices of $D_P$ to some arbitrary vertices of $D_Q$.
    For the remaining $n'$ vertices, since $Q$ is fixed, we can try out all possible mappings from the vertices of $D_P$ to the vertices of $D_Q$. Note that since we have $m$ arcs in $D_P$, there are at most $2m$ non-isolated vertices in $D_P$, and since $m\le \ell$, we have $n'\le 2\ell$. So, there are at most $k^{2\ell}$ potential mappings, which can be considered constant (since $Q$ is fixed). For each mapping, we can check in polynomial time if it is indeed the homomorphism that we are looking for.

    To check this, for a particular mapping $h$, we can again go over all tuples of the form $(p_1,\ldots p_m)$ where each $p_i$ is a path in $Q$ from $h(x_i)$ to $h(y_i)$ and verify that $L(p_i)\subseteq L(E_P(e_i))$ and that there are no arcs that are repeated in $(p_1,\ldots p_m)$. Note that each $p_i$ is a path over arcs with label from $\{a,a^+\}$ and also $e_i$ is an arc with label from $\{a,a^+\}$. Thus, checking whether $L(p_i)\subseteq L(E_P(e_i))$ can be done in polynomial time as follows:
    
    \begin{itemize}
        \item  If $e_i$ has label $a$ then $p_i$ should be a path of length $1$ with edge label $a$. 
        \item If $e_i$ has label $a^+$ then $p_i$ can be any trail/simple path from $h(x_i)$ to $h(y_i)$.
    \end{itemize}

    Note that since we are checking for trails, each path $p_i$ has length at most $\ell$. Moreover, there are at most $2^{\ell}$ many possible candidates for a path $p_i$ between any two vertices in $Q$. Hence, the number of possible tuples of the form $(p_1,\ldots p_m)$ is $m^{2^{\ell^2}}$ and since $m\le \ell$, the overall running time of this algorithm would be $k^{2\ell}\ell^{2^{\ell^2}}$ which is a constant in the size of $P$ for a fixed $Q$. Thus, in total, this is a polynomial-time algorithm.
    \bigskip

    The argument for \textsc{\rgphom[ispn]($Q$)} is analogous, and here, we even have $n\leq k$ by the vertex-injectivity property.
\end{proof}

However, the uniform case is NP-complete. We reduce from \textsc{Clique}.

\decisionpb{\textsc{Clique}}{An undirected graph $G$ and an integer $k$.}{Is there a clique of size $k$ in $G$?}{0.9}

\begin{theorem}
    \textsc{\rgphom[itn]} and \textsc{\rgphom[ispn]} are NP-complete, even for input RGPs with arc-labels all-$a$, where $a\in\Sigma$.
\end{theorem}
\begin{proof}
 We reduce from \textsc{Clique}. Let $(G,k)$ be an input of \textsc{Clique}. We construct the RGP $P$ with $k$ vertices by adding two symmetric arcs labeled $a$ for any two distinct vertices in $P$. The RGP $Q$ is obtained from $G$ by letting $V(D_Q)=V(G)$ and by adding for each edge $\{x,y\}$  of $G$ two symmetric arcs $(x,y)$ and $(y,x)$ labeled $a$. Notice that for $P$ and $Q$ defined as above, since every arc-label is $a$, $P\xrightarrow{ispn}Q$ if and only if $P\xrightarrow{itn}Q$, thus it suffices to prove the reduction for, say, \textsc{\rgphom[itn]}.
 
  If $G$ has a clique of size $k$, we construct a mapping as follows. The vertices of $P$ map to the $k$ vertices of $Q$ corresponding to the $k$-clique of $G$. This is an itn-homomorphism since each arc of $P$ is mapped to a distinct arc of $Q$. 
  
  Conversely, if $P\xrightarrow{itn}Q$, then every vertex of $P$ must be mapped to a distinct vertex of $Q$, since any two vertices of $P$ are joined by two symmetric $a$-arcs and there are no loops in $Q$. Thus, the $k$ vertices in $Q$ are all joined by two symmetric $a$-arcs as well, and so in $G$ they form a $k$-clique.
\end{proof}

\subsection{P/NP-complete dichotomies for non-injective semantics}

We next consider RGPs with arc-labels in $\{a,a^+\}$ for the non-injective semantics, in the non-uniform regime (i.e. fixed target RGPs $Q$). We can derive a CSP-style dichotomy for every RGP $Q$, which can be made explicit for undirected RGPs. The non-uniform (i.e. fixed-template) Constraint Satisfaction Problem (CSP) can be expressed as \textsc{\HOM($T$)} for a fixed template relational structure $T$. It is in fact long known that every such CSP for template $T$ is equivalent to \textsc{\HOM($D_T$}) for a digraph $D_T$~\cite{FV98}. Bulatov and Zhuk independently proved a P/NP-complete dichotomy for CSPs, which can be stated as follows.

\begin{theorem}[Bulatov and Zhuk~\cite{B17,Z17}]\label{thm:CSP-dicho}
    \textsc{\HOM($T$)} is polynomial-time solvable if and only if the core of $T$ has a weak near-unanimity polymorphism of arity at least~3.\footnote{We do not define weak near-unanimity polymorphisms and refer to~\cite{B17,Z17} instead.} Otherwise, \textsc{\HOM($T$)} is NP-complete.
\end{theorem}

We use Theorem~\ref{thm:CSP-dicho} to prove the following.

\begin{theorem}
    For any fixed RGP $Q$ with arc-labels in $\{a,a^+\}$, \textsc{\rgphom[tn]($Q$)}, \textsc{\rgphom[spn]($Q$)} and \textsc{\rgphom($Q$)} (when restricted to input RGPs with arc-labels in $\{a,a^+\}$) are polynomial-time solvable if the core of $D(Q)$ has a weak near-unanimity polymorphism of arity at least~3. Otherwise, they are NP-complete.
 \end{theorem}
\begin{proof}
By Observation~\ref{obs:a-equiv-noninjective}, it suffices to prove the theorem for \textsc{\rgphom($Q$)}.

Let $Q$ be an RGP with arc-labels in $\{a,a^+\}$. We show that \textsc{\rgphom($Q$)} (when restricted to input RGPs with arc-labels in $\{a,a^+\}$) is polynomially equivalent to \textsc{\HOM($D(Q)$)}, meaning that they are polynomially reducible to each other. Thus, if  \textsc{\HOM($D(Q)$)} is NP-complete, then \textsc{\rgphom($Q$)} (when restricted to input RGPs with arc-labels in $\{a,a^+\}$) is NP-complete; if it is polynomial, then \textsc{\rgphom($Q$)} (when restricted to input RGPs with arc-labels in $\{a,a^+\}$) is polynomial. Thus, the statement will then follow from Theorem~\ref{thm:CSP-dicho}.
    
By Proposition~\ref{prop:a-a+-reduction-to-HOM}, \textsc{\rgphom($Q$)} reduces to \textsc{\HOM($D(Q)$)}. It thus suffices to show the converse. Let $D$ be an input digraph for \textsc{\HOM($D(Q)$)}. We may assume that all arc-labels of $D$ are $a$ or $t$, as otherwise it is a trivial NO-instance. Let $P_D$ be the RGP with the same underlying digraph as $D$, where we replace each arc labeled $t$ by an arc labeled $a^+$, and each arc labeled $a$ remains the same. We show that $D\to D(Q)$ if and only if $P_D\nto Q$. If $D\to D(Q)$, we use the same mapping $f$ for $P_D$ to $Q$. For an arc $(x,y)$ of $P_D$, if its label is $a$, $(f(x),f(y))$ has label $a$ in $D(Q)$ and thus also in $Q$ since their $a$-labels are the same. If $(x,y)$ is labeled $a^+$ in $P_D$, then $(x,y)$ was labeled $t$ in $D$, and thus $(f(x),f(y))$ has label $t$ in $D(Q)$. Hence, in $Q$, we have a directed path from $f(x)$ to $f(y)$, and thus $(x,y)$ satisfies the condition for $f$ to be an n-homomorphism of $P_D$ to $Q$. Thus, $P_D\nto Q$. Conversely, assume that $P_D\nto Q$. The argument is exactly the same: the n-homomorphism $f$ witnessing $P_D\nto Q$ is a homomorphism of $D$ to $D(Q)$, as an $a$-arcs of $P_D$ maps to an $a$-arcs of $Q$, and an $a^+$-arc $(x,y)$ of $P_D$ (which is a $t$-arc in $D$) maps to a directed paths from $f(x)$ to $f(y)$ in $Q$, which implies the existence of a $t$-arc $(f(x),f(y))$ in $D(Q)$. This completes the proof.
\end{proof}

We next consider \emph{undirected} RGPs with arc-labels in $\{a,a^+\}$ for the non-injective semantics. We use the well-known dichotomy theorem of Hell and Ne\v{s}et\v{r}il for undirected graph homomorphism problems~\cite{HN90} to obtain an explicit dichotomy condition for \textsc{\rgphom($Q$)} and its non-injective variants. Here, arcs are pairs of vertices, called \emph{edges} (equivalently, for each arc from $x$ to $y$, we have its symmetric arc from $y$ to $x$ with the same label). In an n-homomorphism of an RGP $P$ to an RGP $Q$ that both satisfy this constraint, two vertices $x$ and $y$ joined by an edge labeled $a$ in $P$ must be mapped to two vertices joined by an edge labeled $a$ in $Q$ (as in a classic graph homomorphism). If $x$ and $y$ are joined by an edge labeled $a^+$ in $P$, they simply need to be mapped to two vertices of $Q$ that are connected by some path in $Q$. Thus, this can be seen as an extension of classic graph homomorphisms with additional (binary) connectivity constraints.

\begin{proposition}\label{prop:a-a+ncore}
  Let $Q$ be an undirected connected n-core RGP over alphabet $\{a\}$ with arc-labels in $\{a,a^+\}$, let $\{S_1,\ldots,S_p\}$ be the set of connected components of the sub-RGP $Q^a$ of $Q$ induced by the edges labeled $a$, and denote by $X_i$ the set of vertices of $S_i$ incident with an edge labeled $a^+$. Then, the following properties hold.
  \begin{itemize}

  \item[(a)] If $Q$ contains no edge labeled $a$, then $Q$ has at most one edge.
  \item[(b)] If $Q$ contains a loop at vertex $v$, then $V(Q)=\{v\}$.
  \item[(c)] If $Q$ is loop-free, then every edge labeled $a^+$ is a bridge.
  \end{itemize}
\end{proposition}
\begin{proof}

  \noindent (a) If all edges of $Q$ are labeled $a^+$ (and $Q$ has at least one edge), since $Q$ is connected, $Q$ maps to any of its subgraphs containing exactly one edge, a contradiction.\\

  \noindent (b) If there is a loop at vertex $v$ that is labeled $a$, then $Q$ maps to its subgraph induced by $\{v\}$, and so, $V(Q)=\{v\}$. If the loop is labeled $a^+$ but there are at least two vertices in $Q$, then $Q$ maps to its subgraph with that loop removed, contradicting the fact that $Q$ is an n-core.\\

  \noindent (c) If two distinct vertices $x,y$ of $Q$ are joined by an edge labeled $a^+$ that is not a bridge, then $Q$ maps to its subgraph obtained by removing that edge, a contradiction.
\end{proof}

An example of an undirected n-core RGP is given in Figure~\ref{fig:undirected-n-core}, where the graph induced by the edges labeled $a$ has three components: a single vertex, a triangle and the Gr\"otzsch graph. The two latter ones are not homomorphic to each other.

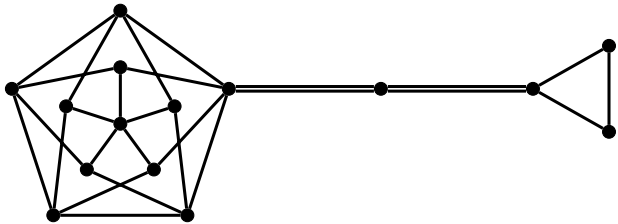
\begin{figure}[!htpb]
  \centering
  \scalebox{1}{\begin{tikzpicture}[join=bevel,inner sep=0.6mm]

\node[draw,shape=circle,fill] (x0) at (18:1.5) {};
\node[draw,shape=circle,fill] (x1) at (90:1.5) {};
\node[draw,shape=circle,fill] (x2) at (162:1.5) {};
\node[draw,shape=circle,fill] (x3) at (234:1.5) {};
\node[draw,shape=circle,fill] (x4) at (306:1.5) {};
\draw[-,line width=0.4mm] (x0)--(x1)--(x2)--(x3)--(x4)--(x0);

\node[draw,shape=circle,fill] (y0) at (18:0.75) {};
\node[draw,shape=circle,fill] (y1) at (90:0.75) {};
\node[draw,shape=circle,fill] (y2) at (162:0.75) {};
\node[draw,shape=circle,fill] (y3) at (234:0.75) {};
\node[draw,shape=circle,fill] (y4) at (306:0.75) {};
\draw[-,line width=0.4mm] (x4)--(y0)--(x1)
                          (x0)--(y1)--(x2)
                          (x1)--(y2)--(x3)
                          (x2)--(y3)--(x4)
                          (x3)--(y4)--(x0);

\node[draw,shape=circle,fill] (0) at (0,0) {};
\draw[-,line width=0.4mm] (y0)--(0)--(y1) (y2)--(0)--(y3)  (y4)--(0);

\node[draw,shape=circle,fill] (p0) at ($(x0)+(2,0)$) {};
\node[draw,shape=circle,fill] (t0) at ($(p0)+(2,0)$) {};
\node[draw,shape=circle,fill] (t1) at ($(t0)+(1,0.57)$) {};
\node[draw,shape=circle,fill] (t2) at ($(t0)+(1,-0.57)$) {};

\draw[-,line width=0.4mm,double] (x0)--(p0)--(t0);

\draw[-,line width=0.4mm] (t0)--(t1)--(t2)--(t0);

  \end{tikzpicture}}
  \caption{An undirected n-core RGP over alphabet $\{a\}$ with edge-labels in $\{a,a^+\}$. Doubled edges are labeled $a^+$, the others are labeled $a$.}
  \label{fig:undirected-n-core}
\end{figure}

The following classic dichotomy result for \textsc{\HOM($H$)} problems will be used to prove the main theorem of this subsection.

\begin{theorem}[Hell and Ne\v{s}et\v{r}il \cite{HN90}]\label{thm:Hom(H)-dicho}
For any undirected graph $H$, \textsc{\HOM($H$)} is polynomial if $H$ is bipartite or contains a loop, and NP-complete otherwise.
\end{theorem}

\begin{theorem}\label{thm:undirected-dicho}
Let $Q$ be an undirected and connected n-core RGP over alphabet $\{a\}$ with edge-labels in $\{a,a^+\}$. If $Q$ has at most one edge, \textsc{\rgphom[tn]($Q$)}, \textsc{\rgphom[spn]($Q$)} and \textsc{\rgphom($Q$)} are solvable in polynomial time. Otherwise, \textsc{\rgphom[tn]($Q$)}, \textsc{\rgphom[spn]($Q$)} and \textsc{\rgphom($Q$)} are NP-complete.
\end{theorem}
\begin{proof}  

By Observation~\ref{obs:a-equiv-noninjective}, it suffices to prove the theorem for \textsc{\rgphom($Q$)}.

 \medskip

  \noindent\emph{Polynomial-time part.} Suppose that $Q$ has at most one edge. Then, it either consists of a single vertex with at most one loop, or of two vertices joined by a unique edge. Let $P$ be an input RGP over alphabet $\{a\}$ and with arc-labels in $\{a,a^+\}$. If $Q$ has one vertex and no loop, $P\nto Q$ if and only if $P$ has no edge. If $Q$ has one vertex and a loop labeled $a$, always $P\nto Q$. If $Q$ has one vertex and a loop labeled $a^+$, $P\nto Q$ if and only if $P$ has no edge labeled $a$. If $Q$ has two vertices and an edge labeled $a^+$, also $P\nto Q$ if and only if $P$ has no edge labeled $a$. Finally, if $Q$ has two vertices and an edge labeled $a$, $P\nto Q$ if and only if the subgraph of $P$ induced by the edges labeled $a$ is bipartite. All these conditions can be checked in polynomial time.

  \medskip

  \noindent\emph{NP-complete part.} Suppose now that $Q$ has at least two edges. \textsc{\rgphom($Q$)} is in NP by Proposition~\ref{prop:a-a+-NP}. We reduce from \textsc{\HOM($H$)}, where $H$ is the undirected graph obtained from $Q^a$ (the sub-RGP of $Q$ induced by the edges labeled $a$) by removing all edge-labels. By Proposition~\ref{prop:a-a+ncore}(a)--(b), $H$ has no loop and at least one edge $xy$. If $H$ was bipartite, then $Q$ would map to $Q[\{x,y\}]$, a contradiction to the fact that $Q$ is an n-core with at least two edges. Hence, $H$ has a connected component with an odd cycle, and thus by Theorem~\ref{thm:Hom(H)-dicho}, \textsc{\HOM($H$)} is NP-complete. We trivially reduce \textsc{\HOM($H$)} to \textsc{\rgphom($Q$)} as follows: for an input undirected graph $G$ of \textsc{\HOM($H$)}, construct the undirected RGP $P(G)$ from $G$ by labeling each edge with $a$. Now, $G\to H$ if and only if $P(G)\nto Q$. So, \textsc{\rgphom($Q$)} is NP-complete.
\end{proof}

\subsection{Directed path RGPs with all arcs labeled $a$ (all semantics)}

In this section, we consider \textsc{\rgphom} when the target RGP $Q$ is a directed path whose arc labels are all $a$ (arguably the simplest RGP directed graph example) and where instances can have labels in $\{a,a^+\}$. This case turns out to have an interesting connection to the following parallel job scheduling problem, studied in the book~\cite[Chapter 4.4, p.666]{algoBook}.

\decisionpb{\textsc{Parallel Job Scheduling With Relative Deadlines}}{A set $J$ of jobs, a duration function $d:J\to\mathbb{N}$, a relative deadline function $r:J\times J\to\mathbb{Z}$, and a maximum time $t_{max}$.}{Is there a feasible schedule for the jobs, that is, an assignment $t:J\to \mathbb N$ of start times such that every job finishes at most at time $t_{max}$ and for any two jobs $j_1$ and $j_2$, $j_1$ starts no later than time $t(j_2)+r(j_1,j_2)$?}{0.9}

Note that, in the above definition, the relative deadlines can be negative, and thus, the problem generalizes the classic scheduling problem with simple precedence constraints (one can ask that $j_1$ precedes $j_2$ by setting $r(j_1,j_2)$ to $-d(j_1)$). \textsc{Parallel Job Scheduling With Relative Deadlines} can be solved in polynomial time by a reduction to a shortest path problem in an arc-weighted digraph, see~\cite[Chapter 4.4, p.666]{algoBook}. We will now show that when $Q$ is a directed path whose arc labels are in $\{a,a^+\}$, \textsc{\rgphom($Q$)} can be modeled as such a scheduling problem.

We say that a digraph is \emph{layered} if it is acyclic and its vertices can be partitioned into sets (called \emph{layers}) $L_1,\ldots,L_k$ such that for each arc $(x,y)$ with $x\in L_i$, we have $y\in L_{i+1}$ \cite[Chapter 9]{layered}. One can check that a digraph is layered if and only if it has a homomorphism to a directed path (and the length of this path must be at least the number of layers minus one).

\begin{theorem}\label{thm:directed-paths-a-a+}
\textsc{\rgphom[tn]}, \textsc{\rgphom[spn]} and \textsc{\rgphom} for instances where $P$ has arc-labels in $\{a,a^+\}$ and the target RGP $Q$ has as its underlying digraph a directed path whose arc-labels are all $a$, can be reduced in polynomial time to \textsc{Parallel Job Scheduling With Relative Deadlines}, and can thus be solved in polynomial time.
\end{theorem}
\begin{proof}
By Observation~\ref{obs:a-equiv-noninjective}, it suffices to prove the theorem for \textsc{\rgphom}.

Let $q_0,\ldots,q_{n-1}$ be the vertices of $Q$, with the arc $(q_i,q_{i+1})$ whenever $0\leq i\leq n-2$. Let $P$ be an instance of \textsc{\rgphom($Q$)}. It is clear that the underlying digraph of $P$ must be acyclic, otherwise $P$ is a NO instance. Now, consider the digraphs $P^A$ and $Q^A$ (note that $Q^A = Q$), which are the subdigraphs of the underlying digraphs of $P$ and $Q$ respectively, that are induced by the arcs labeled $a$. In order for $P$ to be a YES-instance, we must necessarily have that $P^A$ has a homomorphism to $Q^A$, that is, $P^A$ must be a layered digraph, with at most as many layers as the number of vertices in $Q$.

  We will now transform $P$ into an equivalent instance $P'$. Observe that in an n-homomorphism of $P$ to $Q$, necessarily all vertices of $P$ that belong to a same connected component of $P^A$ and to the same layer in $P^A$ will be mapped to the same vertex of $Q$. Thus, we may identify all such vertices of $P$ into a single vertex to obtain $P'$. In $P'$, all connected components of $P'^A$ are directed paths. Now, observe that if there is some arc labeled $a^+$ within some connected component of $P'^A$, either it creates a directed cycle (then $P$ is a NO instance and we return NO), or it is not useful since this arc can be mapped trivially. Thus, we may assume that there is no such arc in $P'$. See Figure~\ref{fig:path-proof1} for an illustration of this process.

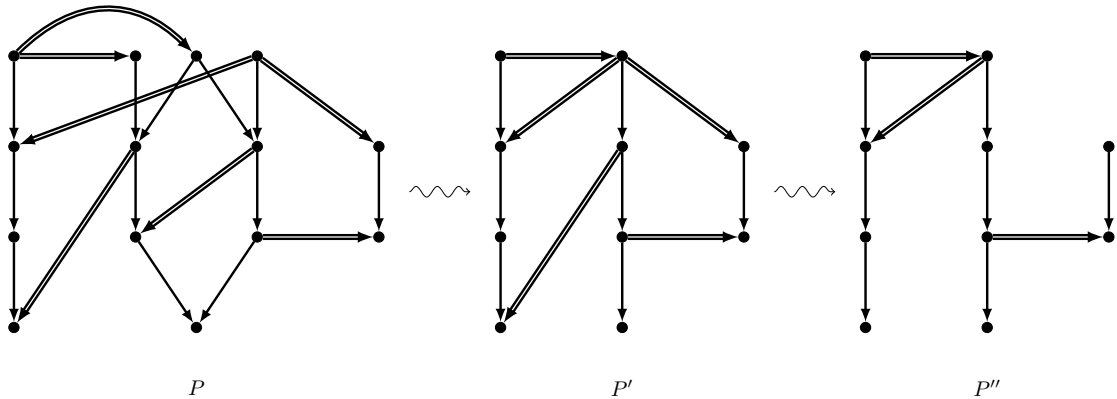
\begin{figure}[!htpb]
  \centering
  \scalebox{0.8}{\begin{tikzpicture}[join=bevel,inner sep=0.6mm]
      \begin{scope}
        \node[draw,shape=circle,fill] (x0) at (0,0) {};
        \draw (x0)+(0,1.5) node[draw,shape=circle,fill] (x1) {};
        \draw (x1)+(0,1.5) node[draw,shape=circle,fill] (x2) {};
        \draw (x2)+(0,1.5) node[draw,shape=circle,fill] (x3) {};
        
        \node[draw,shape=circle,fill] (y0) at (3,0) {};
        \draw (y0)+(-1,1.5) node[draw,shape=circle,fill] (y1) {};
        \draw (y1)+(0,1.5) node[draw,shape=circle,fill] (y2) {};
        \draw (y2)+(0,1.5) node[draw,shape=circle,fill] (y3) {};
        \draw (y0)+(1,1.5) node[draw,shape=circle,fill] (y4) {};
        \draw (y4)+(0,1.5) node[draw,shape=circle,fill] (y5) {};
        \draw (y5)+(-1,1.5) node[draw,shape=circle,fill] (y6) {};
        \draw (y5)+(0,1.5) node[draw,shape=circle,fill] (y7) {};
        
        \node[draw,shape=circle,fill] (z0) at (6,1.5) {};
        \draw (z0)+(0,1.5) node[draw,shape=circle,fill] (z1) {};     

        \node (Q) at (3,-1) {$P$};      
        
        \draw[->,>=latex, double ,line width=0.4mm] (y7)--(x2);
        \draw[->,>=latex, line width=0.4mm] (x3)--(x2);
        \draw[->,>=latex, line width=0.4mm] (x2)--(x1);
        \draw[->,>=latex, line width=0.4mm] (x1)--(x0);
        \draw[->,>=latex, line width=0.4mm] (y2)--(y1);
        \draw[->,>=latex, line width=0.4mm] (y1)--(y0);
        \draw[->,>=latex, line width=0.4mm] (y3)--(y2);
        \draw[->,>=latex, line width=0.4mm] (y6)--(y2);
        \draw[->,>=latex, line width=0.4mm] (y6)--(y5);
        \draw[->,>=latex, line width=0.4mm] (y7)--(y5);
        \draw[->,>=latex, line width=0.4mm] (y5)--(y4);
        \draw[->,>=latex, line width=0.4mm] (y4)--(y0);
        \draw[->,>=latex, line width=0.4mm] (z1)--(z0);
        \draw[->,>=latex, double ,line width=0.4mm] (y2)--(x0);
        \draw[->,>=latex, double ,line width=0.4mm] (x3)--(y3);
        \draw[->,>=latex, double ,line width=0.4mm] (x3) .. controls +(1,1) and +(-1,1) .. (y6);
        \draw[->,>=latex, double ,line width=0.4mm] (y5)--(y1);
        \draw[->,>=latex, double ,line width=0.4mm] (y7)--(z1);
 %       \draw[->,>=latex, double ,line width=0.4mm] (y7)--(x2);
        \draw[->,>=latex, double ,line width=0.4mm] (y4)--(z0);

        \draw[decorate, decoration={snake},->] (6.5,2.25) -- (7.5,2.25);
      \end{scope}

      \begin{scope}[xshift=8cm]      
        \node[draw,shape=circle,fill] (x0) at (0,0) {};
        \draw (x0)+(0,1.5) node[draw,shape=circle,fill] (x1) {};
        \draw (x1)+(0,1.5) node[draw,shape=circle,fill] (x2) {};
        \draw (x2)+(0,1.5) node[draw,shape=circle,fill] (x3) {};

        \node[draw,shape=circle,fill] (y0) at (2,0) {};
        \draw (y0)+(0,1.5) node[draw,shape=circle,fill] (y1) {};
        \draw (y1)+(0,1.5) node[draw,shape=circle,fill] (y2) {};
        \draw (y2)+(0,1.5) node[draw,shape=circle,fill] (y3) {};
                
        \node[draw,shape=circle,fill] (z0) at (4,1.5) {};
        \draw (z0)+(0,1.5) node[draw,shape=circle,fill] (z1) {};     

        \node (Q) at (2,-1) {$P'$};      
        
        \draw[->,>=latex, line width=0.4mm] (x3)--(x2);
        \draw[->,>=latex, line width=0.4mm] (x2)--(x1);
        \draw[->,>=latex, line width=0.4mm] (x1)--(x0);
        \draw[->,>=latex, line width=0.4mm] (y3)--(y2);
        \draw[->,>=latex, line width=0.4mm] (y2)--(y1);
        \draw[->,>=latex, line width=0.4mm] (y1)--(y0);
        \draw[->,>=latex, line width=0.4mm] (z1)--(z0);
        \draw[->,>=latex, double ,line width=0.4mm] (y2)--(x0);
        \draw[->,>=latex, double ,line width=0.4mm] (x3)--(y3);
        \draw[->,>=latex, double ,line width=0.4mm] (y3)--(z1);
        \draw[->,>=latex, double ,line width=0.4mm] (y1)--(z0);
        \draw[->,>=latex, double ,line width=0.4mm] (y3)--(x2);

        \draw[decorate, decoration={snake},->] (4.5,2.25) -- (5.5,2.25);
      \end{scope}

      \begin{scope}[xshift=14cm]      
        \node[draw,shape=circle,fill] (x0) at (0,0) {};
        \draw (x0)+(0,1.5) node[draw,shape=circle,fill] (x1) {};
        \draw (x1)+(0,1.5) node[draw,shape=circle,fill] (x2) {};
        \draw (x2)+(0,1.5) node[draw,shape=circle,fill] (x3) {};

        \node[draw,shape=circle,fill] (y0) at (2,0) {};
        \draw (y0)+(0,1.5) node[draw,shape=circle,fill] (y1) {};
        \draw (y1)+(0,1.5) node[draw,shape=circle,fill] (y2) {};
        \draw (y2)+(0,1.5) node[draw,shape=circle,fill] (y3) {};
                
        \node[draw,shape=circle,fill] (z0) at (4,1.5) {};
        \draw (z0)+(0,1.5) node[draw,shape=circle,fill] (z1) {};     

        \node (Q) at (2,-1) {$P''$};      
        
        \draw[->,>=latex, line width=0.4mm] (x3)--(x2);
        \draw[->,>=latex, line width=0.4mm] (x2)--(x1);
        \draw[->,>=latex, line width=0.4mm] (x1)--(x0);
        \draw[->,>=latex, line width=0.4mm] (y3)--(y2);
        \draw[->,>=latex, line width=0.4mm] (y2)--(y1);
        \draw[->,>=latex, line width=0.4mm] (y1)--(y0);
        \draw[->,>=latex, line width=0.4mm] (z1)--(z0);
%        \draw[->,>=latex, double ,line width=0.4mm] (y2)--(x0);
        \draw[->,>=latex, double ,line width=0.4mm] (x3)--(y3);
        \draw[->,>=latex, double ,line width=0.4mm] (y1)--(z0);
        \draw[->,>=latex, double ,line width=0.4mm] (y3)--(x2);
      \end{scope}      
    \end{tikzpicture}}
      \caption{Proof of Theorem~\ref{thm:directed-paths-a-a+}: illustration of the two-step simplification of an input RGP $P$ into $P'$ and $P''$, with three components in $P^A$. Doubled arcs are labeled $a^+$, the others are labeled $a$.}
  \label{fig:path-proof1}
\end{figure}

  We now perform a second simplification and transform $P'$ into yet another equivalent RGP $P''$. Consider two components $C_1$ and $C_2$ of $P'^A$. We will show that, if there are two $a^+$-labeled arcs from $C_1$ to $C_2$, then one of them can be removed to obtain an equivalent instance. Recall that each of $C_1$ and $C_2$ induces a directed path, where the vertices of $C_1$ are $u_1,u_2,\ldots,u_{|C_1|}$ and those of $C_2$ are $v_1,v_2,\ldots,v_{|C_2|}$ (where $(u_i,u_{i+1})$ and $(v_i,v_{i+1})$ are $a$-labeled arcs for $i$ less than $|C_1|$ and $|C_2|$, respectively). For any $a^+$-labeled arc $(u_i,v_j)$, the image of $u_i$ by the homomorphism needs a directed path (of length at least~1) to the image of $v_j$ in $Q$. Assume $u_1$ is mapped to $q_r$ in $Q$, and $v_1$ is mapped to $q_s$ in $Q$; thus, $u_i$ is mapped to $q_{r'}$ with $r'=r+i-1$ and $v_j$ is mapped to $q_{s'}$ with $s'=s+j-1$. We can now express the constraint given by the $a^+$-labeled arc $(u_i,v_j)$ using the images of $u_1$ and $v_1$ only. Indeed, we need $s'>r'$; that is, $s+j-1>r+i-1$, or $s\geq r+i-j+1$.
  If this condition is not fulfilled, the arc from $u_i$ to $v_j$ cannot be accommodated by the homomorphism. Thus, considering all such arcs $(u_i,v_j)$ in $P'$, we may remove all of them but the one that maximizes $i-j$, indeed it is the strongest restriction among them. Repeating this for all pairs of components, we obtain the equivalent instance $P''$. See again Figure~\ref{fig:path-proof1} for an illustration.

We can now build the equivalent instance of \textsc{Parallel Job Scheduling With Relative Deadlines}. We let $J$ be the set of connected components of $P''^A$. We let $t_{max}=n-1$. For a connected component $C$ of $P''^A$, $d(C)$ (recall that $C$ is a directed path) is the number of arcs of $C$. For two components $C_1$ and $C_2$ with vertices $u_1,u_2,\ldots,u_{|C_1|}$ and $v_1,v_2,\ldots,v_{|C_2|}$ numbered as before, if there is an $a^+$-labeled arc $(u_i,v_j)$ from $C_1$ to $C_2$ in $P''$ (recall that there is at most one such arc), we let $r(C_1,C_2)=j-i-1$. That is, if $j-i-1\geq 0$, job $C_1$ must start at most $j-i-1$ time units after job $C_2$; if $j-i-1<0$, then $C_1$ must start at least $i-j+1$ time units before $C_2$. If there is no arc from $C_1$ to $C_2$, then we let $r(C_1,C_2)=n$, which imposes no time constraint.

It is now clear that any solution $s$ to this instance of \textsc{Parallel Job Scheduling With Relative Deadlines} corresponds to an n-homomorphism $f_s$ of $P''$ to $Q$ (and thus of $P$ to $Q$), and vice-versa, where the start time of a job $C$ in $s$ is equal to the index of the image by $f_s$ of the source of $C$. 
\end{proof}

We next show that the injective variants have an even easier polynomial-time algorithm.

\begin{theorem}
    \textsc{\rgphom[itn]} and \textsc{\rgphom[ispn]} for instances where $P$ has arc-labels in $\{a,a^+\}$ and the target RGP $Q$ has as its underlying digraph a directed path all whose arc-labels are $a$, can be solved in polynomial time.
\end{theorem}
\begin{proof}
    Since we are in the injective semantics and $D_Q$ is a directed path, notice that any two vertices of $D_P$ with an outgoing arc cannot be mapped to the same vertex of $D_Q$, and similarly, no two vertices with an incoming arc can be mapped to the same vertex of $D_Q$. Moreover, in order to map to $D_Q$, $D_P$ must be acyclic. Hence, only RGPs $P$ whose underlying digraph $D_P$ is a collection of directed paths can be mapped to $Q$.

\medskip
    
    Now, for \textsc{\rgphom[itn]}, such an RGP $P$ can be mapped to $Q$ if and only if the total number of arcs in $D_P$ is at most the number of arcs in $D_Q$.
    Indeed, by the injectivity conditions, no two arcs of $D_P$ can use the same arc of $D_Q$ in the path witnessing their mapping, and any arc of $D_P$ needs at least one arc of $D_Q$, thus the number of arcs in $D_Q$ is at least that of $D_P$. On the other hand, if this condition is satisfied, one can simply greedily map the paths in $D_P$ one by one, starting from the first vertex of $D_Q$, in a natural way. The last vertex of a path in $D_P$ is mapped to the same vertex of $D_Q$ as the first vertex of the next path in $D_P$. (And arcs labeled $a^+$ are simply mapped to single arcs.)

      Since one can check in polynomial time whether $D_P$ is a collection of disjoint directed paths and has at most as many arcs as $D_Q$, the result follows. 

      \medskip

     Similarly, for \textsc{\rgphom[ispn]}, if $P$ is an RGP where $D_P$ is a collection of $k$ directed paths of lengths $l_1,\ldots l_k$, then $P\xrightarrow{ispn}Q$ if and only if $D_Q$ has $k$ vertex-disjoint paths of lengths at least $l_1,\ldots l_k$ and $D_Q$ contains at least as many vertices as $D_P$, which is true if and only if the number of vertices in $D_Q$ is at least the number of vertices in $D_P$. Again, this can be checked in polynomial time.
\end{proof}

\subsection{Directed path RGPs with arc-labels in $\{a,a^+\}$ (all semantics)}

In our next result, we use a stronger method to show a stronger version of Theorem~\ref{thm:directed-paths-a-a+}. This also extends a result from~\cite{MS04}, where the statement is proved only for input RGPs whose underlying digraphs are directed trees. Here we prove it for all kinds of inputs.

But first, we need to define some notions that are classic useful tools for proving algorithmic results for homomorphism problems. For an arc-labeled digraph $D$ and a positive integer $k$, we define the \emph{product digraph} $D^k$ as the digraph on vertex set $V(D)^k$, with an arc of label $\ell$ from $(x_1,\ldots,x_k)$ to $(y_1,\ldots,y_k)$ if all pairs $(x_i,y_i)$ with $1\leq i\leq k$ are arcs of label $\ell$ in $D$.

A homomorphism of $D^k$ to $D$ is called a ($k$-ary) \emph{polymorphism} of $D$. For a set $S$, a function $f$ from $S^3$ to $S$ is a \emph{majority function} if for all $x,y$ in $S$, $f(x,x,y)=f(x,y,x)=f(y,x,x)=x$. We have the following theorem, that can be found in \cite[Theorem 5.2.4]{majority} and uses a classic arc-consistency algorithm.

\begin{theorem}[\cite{majority}]\label{thm:majority}
\HOM\ is polynomial-time solvable whenever the target arc-labeled digraph $D$ has a ternary polymorphism that is a majority function.
\end{theorem}

\begin{theorem}\label{thm:a-a+-paths}
 \textsc{\rgphom[tn]}, \textsc{\rgphom[spn]} and \textsc{\rgphom} can be solved in polynomial time whenever the target $Q$ is an RGP whose underlying digraph is a directed path with arc labels in $\{a,a^+\}$.
\end{theorem}
\begin{proof} By Observation~\ref{obs:a-equiv-noninjective}, it suffices to show the statement for \textsc{\rgphom}.

By Proposition~\ref{prop:a-a+-reduction-to-HOM}, it suffices to show that \HOM, for the target digraph $D(Q)$, as defined before Proposition~\ref{prop:a-a+-reduction-to-HOM}, is polynomial-time solvable. To do this, we show that $D(Q)$ admits a majority polymorphism $m:V(D(Q))^3\to V(D(Q))$, and we will obtain the claimed result by Theorem~\ref{thm:majority}.

As previously, let $q_0,\ldots,q_{n-1}$ be the vertices of $Q$, with the arc $(q_i,q_{i+1})$ whenever $0\leq i\leq n-2$. This induces a natural ordering of the vertices with $q_i\leq q_j$ if and only if $i\leq j$. The majority polymorphism $m$ we consider is the median: for three (not necessarily distinct) integers $i,j,k$ with $0\leq i\leq j\leq k\leq n-1$ and $\{a,b,c\}=\{i,j,k\}$, we let $m(q_a,q_b,q_c)=q_j$.

First, it is clear that $f$ is a majority function by our definition since the median of $q_i$, $q_i$ and $q_j$ is always $q_i$.

It remains to show that $f$ is a polymorphism, that is, a homomorphism of $D(Q)^3$ to $D(Q)$. Let $(q_a,q_b,q_c)$ and $(q_d,q_e,q_f)$ be two vertices of $D(Q)^3$ such that $((q_a,q_b,q_c),(q_d,q_e,q_f))$ is an arc in $D(Q)^3$ (that is, $(q_a,q_d)$, $(q_b,q_e)$ and $(q_c,q_f)$ are arcs in $D(Q)$ with the same label). Assume without loss of generality that $a\leq b\leq c$. Since $D(Q)$ has two arc labels, there are two cases to consider.

If the arc is labeled $a$, we necessarily have $d=a+1$, $e=b+1$ and $f=c+1$; since  $a\leq b\leq c$, also $d\leq e\leq f$. Thus, we have $m(q_a,q_b,q_c)=q_b$ and $m(q_d,q_e,q_f)=q_e=q_{b+1}$. Thus there is an arc labeled $a$ from $m(q_a,q_b,q_c)$ to $m(q_d,q_e,q_f)$, as required.

If the arc is labeled ``$t$'', we have $d\geq a+1$, $e\geq b+1$ and $f\geq c+1$. Since $a\leq b\leq c$, we will have $m(q_d,q_e,q_f)=q_m$ with $m\geq b+1$. Since $m(q_a,q_b,q_c)=q_b$, there is an arc with label ``$t$'' from $m(q_a,q_b,q_c)$ to $m(q_d,q_e,q_f)$, as required. This completes the proof.
\end{proof}

We remark that Theorem~\ref{thm:a-a+-paths} also applies to RGPs with vertex-labels (where a vertex with a given label can only be assigned to a vertex with the same label). Indeed, a vertex-label is modeled as a unary relation, and unary relations trivially satisfy the properties for having a majority polymorphism.

Moreover, using the same method, Theorem~\ref{thm:a-a+-paths} extends to labels of the form $a^\ast$, $a^k$ or $a^{\leq k}$ where $k\in\mathbb{N}$.

We next show that, on the other hand, the injective variants are more complex even in this simple setting. We recall the classic (strongly) NP-complete \textsc{Bin Packing} problem~\cite[Problem SR1]{GJ79}, as well as its (strongly) NP-complete variant \textsc{3-Partition}~\cite[Problem SP15]{GJ79} (strongly NP-complete means it is NP-complete even when the input is encoded in unary).

\decisionpb{\textsc{Unary Encoded Bin Packing}}{A set of $t$ bins $b_1,\ldots,b_t$, each of capacity $c\in\mathbb{N}$, and a set of $k$ objects $o_1,\ldots,o_k$, with each object $o_i$ of size $s_i\in\mathbb{N}$ (all numbers encoded in unary).}{Is there an assignment $f:[k]\to [t]$ of the objects so that the sum of objects assigned to each bin is at most $c$?}{0.9}

\decisionpb{\textsc{Unary Encoded 3-Partition}}{A set of $3k$ numbers $n_1,\ldots,n_{3k}$ (encoded in unary) whose sum is $N$, with $N/4k<n_i<N/2k$ for each $i\in [3k]$.}{Is there an assignment $f:[3k]\to [k]$ of the numbers so that for any integer $i$ in $[k]$, there are exactly three distinct numbers $p,q,r$ in $[3k]$ assigned to $i$ by $f$, and $n_p+n_q+n_r=\frac{N}{k}$?}{0.9}

\begin{theorem}
    \textsc{\rgphom[itn]} and \textsc{\rgphom[ispn]} are NP-complete when the target RGP $Q$ is a directed path with labels in $\{a,a^+\}$.
\end{theorem}
\begin{proof}
The problems are in NP by Proposition~\ref{prop:a-a+-NP}.

We first show NP-hardness of \textsc{\rgphom[itn]} by reducing from \textsc{Unary Encoded Bin Packing}. Given a \textsc{Unary Encoded Bin Packing} instance, we construct $P$ and $Q$ as follows. $D_P$ consists of the disjoint union of directed paths with all arcs labeled $a$, each corresponding to one object $o_i$ and having $s_i$ arcs. $D_Q$ consists of a directed path of length $t(c+1)$, which alternates $t$ times the following pattern: $c$ arcs labeled $a$ followed by one arc labeled $a^+$.

    If the instance of \textsc{Unary Encoded Bin Packing} is positive, we have a homomorphism as follows. For each bin $b_i$, consider all the objects assigned to $b_i$, and map the corresponding paths in $D_P$ to the $i$th subpath in $D_Q$ with arc-labels $a$, in a consecutive fashion (the last vertex of one path being mapped to the same vertex of $D_Q$ as the first vertex of the next path). This is a homomorphism that respects the injectivity constraints, since each arc in $D_P$ is mapped to a distinct arc of $D_Q$.

\medskip 
    Conversely, if there is an injective mapping of $P$ to $Q$, notice that no two pairs $\{u_1,v_1\}$ and $\{u_2,v_2\}$ of adjacent vertices of $D_P$ can be mapped to the same two vertices $x,y$ of $D_Q$, because of the injectivity constraint. Indeed, in the path $D_Q$, there is a unique path between $x$ and $y$, so this would imply that the two paths witnessing the mapping for the arcs $(u_1,v_1)$ and $(u_2,v_2)$ are equal (and thus have a common arc), contradicting the injectivity of the mapping. Moreover, no arc of $D_P$ can be mapped to an arc labeled $a^+$, since all arcs of $D_P$ are labeled $a$.
    
    Consider then, for every $i$ with $1\leq i\leq t$, the $i$th maximal subpath $D_i$ in $D_Q$ with arc-labels $a$, which has length $c$. Let $S_i$ be the set of path components of $D_P$ mapped to $D_i$. Since $D_i$ is a directed path, each path in $S_i$ of length $\ell$ must be mapped to a subpath of $D_i$ of length $\ell$. As we have already argued, no two arcs of the images overlap. So, together the paths in $S_i$ have at most $c$ arcs. The objects corresponding to the path components in $S_i$ can be assigned to bin $b_i$, since their total size is at most $c$. Hence, we have a positive instance of \textsc{Unary Encoded Bin Packing}. This completes the proof for \textsc{\rgphom[itn]}.

\bigskip

    To show NP-hardness of \textsc{\rgphom[ispn]}, we give a very similar reduction from \textsc{Unary Encoded 3-Partition}. Given a \textsc{Unary Encoded 3-Partition} instance, we construct $P$ and $Q$ as follows. $D_P$ consists of the disjoint union of directed paths with all arcs labeled $a$, each corresponding to one number $n_i$ and having $n_i$ arcs. $D_Q$ consists of a directed path of length $N+3k$, which alternates $k$ times the following pattern: $\frac{N}{k}+2$ arcs labeled $a$ followed by one arc labeled $a^+$.

    We omit the proof, as the argument is analogous to \textsc{\rgphom[itn]} reduction.
\end{proof}

\section{Conclusion}\label{sec:conclu}

We have seen that \textsc{\rgphom}, which is generally PSPACE-hard (but in NP when the target RGP is a graph database), is in EXPSPACE. This matches the complexity of general RGP query containment, which is EXPSPACE-complete~\cite{CDLV00,FLS98}. 
It remains to close the gap between the PSPACE lower bound and the EXPSPACE upper bound. Could it be that \textsc{\rgphom} has actually a lower complexity compared with general RGP query containment? Note that, in our algorithm, the main bottleneck is the length of the walk witnessing that a regular expression matches a given pair of vertices, in Lemma~\ref{lemma:path-length-bound}. Can this upper bound on the walk length be reduced, or are there examples where this length is necessary.

We have also seen that when all labels are in $\{a,a^+\}$ (a case that is also in NP, and that corresponds to XPath and SPARQL queries), we have a complete classification of the NP-complete and polynomial cases for undirected RGPs, and all RGPs whose underlying digraph is a directed path are polynomial. It was proved in~\cite{MS04} that when both the input and target is a directed tree with arc-labels in $\{a,a^+\}$, \textsc{\rgphom} is polynomial-time solvable. Is it true that (for general instances) \textsc{\rgphom($Q$)} is polynomial when $Q$ is a directed tree with arc-labels in $\{a,a^+\}$? \footnote{Note that this would not be true for all acyclic RGPs, indeed there exists an oriented tree $T$ with 20 vertices such that \HOM($T$) is NP-complete~\cite{BBSW23}. Thus, for the RGP $Q(T)$ with $T$ as its underlying digraph and all arc-labels equal to $a$, \textsc{\rgphom($Q(T)$)} is NP-complete.} When all arc-labels are $a$, then the only n-core RGPs whose underlying digraphs are rooted directed trees are, in fact, the directed paths (such a directed tree maps to its longest directed path). But there are many more n-cores when both labels $a$ and $a^+$ are used, see Figure~\ref{fig:directed-tree-ncore} for a simple example.

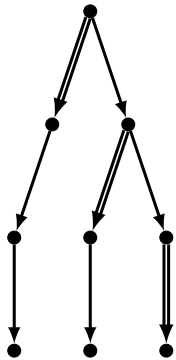
\begin{figure}[!htpb]
  \centering
  \scalebox{1}{\begin{tikzpicture}[join=bevel,inner sep=0.6mm]

      \node[draw,shape=circle,fill] (x) at (0,3) {};
      \node[draw,shape=circle,fill] (y) at (-0.5,1.5) {};
      \node[draw,shape=circle,fill] (z) at (0.5,1.5) {};
      \node[draw,shape=circle,fill] (u) at (-1,0) {};
      \node[draw,shape=circle,fill] (v) at (1,0) {};
      \node[draw,shape=circle,fill] (w) at (0,0) {};
      \node[draw,shape=circle,fill] (s) at (-1,-1.5) {};
      \node[draw,shape=circle,fill] (t) at (1,-1.5) {};
      \node[draw,shape=circle,fill] (p) at (0,-1.5) {};
      
      \draw[->,>=latex, double, line width=0.4mm] (x)--(y);
      \draw[->,>=latex, line width=0.4mm] (u)--(s);
      \draw[->,>=latex, line width=0.4mm] (z)--(v);
      \draw[->,>=latex, double, line width=0.4mm] (z)--(w);
      \draw[->,>=latex, line width=0.4mm] (w)--(p);
      \draw[->,>=latex, double, line width=0.4mm] (v)--(t);
      \draw[->,>=latex, line width=0.4mm] (y)--(u);
      \draw[->,>=latex, line width=0.4mm] (x)--(z);

    \end{tikzpicture}}
      \caption{An n-core directed tree RGP with arc-labels in $\{a,a^+\}$. Doubled arcs are labeled $a^+$, the others are labeled $a$.}
  \label{fig:directed-tree-ncore}
\end{figure}

  It would be interesting to study further cases corresponding to relevant applications, such as the ones of ``simple regular expressions'' studied in~\cite{FGKMNT20}. An interesting case with a unary alphabet is the case where all labels are in $\{a^k | k\in\mathbb{N}\}$. This models a kind of weighted homomorphism, where an arc with weight $k$ (that is, with label $a^k$) can only map to a directed walk with sum of weights equal to $k$. If we consider weights of the form $\{a^{\leq k} | k\in\mathbb{N}\}$, this would be close to the setting of homomorphisms to \emph{powers of digraphs}, also related to the notion of \emph{dilation} of graph embeddings~\cite{powers}.

\subsection*{Acknowledgments} We thank the anonymous referees for carefully reading the manuscript, and for pointing out a serious error in the initial algorithm proving our main theorem. As a result, the wrong EXPTIME algorithm we claimed in the conference version of this paper~\cite{confversion} is now replaced with the current EXPSPACE result. We do not know whether this can be improved to EXPTIME or even PSPACE.

\subsection*{AI usage} No AI-based system was used during the research and writing process of this paper. We used Claude Opus 5.5 to provide an independent review of the paper revision, and it pointed out a few minor mistakes and inaccuracies, which we corrected.

\end{document}